\documentclass[11pt]{article}%
\usepackage{amsfonts}
\usepackage{amsmath}
\usepackage{hyperref}
\usepackage{fullpage}
\usepackage{amssymb}
\usepackage{graphicx}%
\providecommand{\U}[1]{\protect\rule{.1in}{.1in}}
\newtheorem{theorem}{Theorem}

\newtheorem{corollary}[theorem]{Corollary}

\newtheorem{definition}[theorem]{Definition}

\newtheorem{lemma}[theorem]{Lemma}

\newtheorem{proposition}[theorem]{Proposition}
\newtheorem{remark}[theorem]{Remark}

\newenvironment{proof}[1][Proof]{\noindent\textbf{#1.} }{\ \rule{0.5em}{0.5em}}
\numberwithin{equation}{section}
\begin{document}

\title{\textbf{Fidelity of recovery, geometric squashed entanglement, and measurement
recoverability}}
\author{Kaushik P. Seshadreesan\thanks{Hearne Institute for Theoretical Physics,
Department of Physics and Astronomy, Louisiana State University, Baton Rouge,
Louisiana 70803, USA}
\and Mark M. Wilde\footnotemark[1] \thanks{Center for Computation and Technology,
Louisiana State University, Baton Rouge, Louisiana 70803, USA}}
\maketitle

\begin{abstract}
This paper defines the fidelity of recovery of a tripartite quantum state on
systems $A$, $B$, and~$C$ as a measure of how well one can recover the full
state on all three systems if system $A$ is lost and a recovery operation is
performed on system $C$ alone. The surprisal of the fidelity of recovery (its
negative logarithm) is an information quantity which obeys nearly all of the
properties of the conditional quantum mutual information $I(A;B|C)$, including
non-negativity, monotonicity with respect to local operations, duality,
invariance with respect to local isometries, a dimension bound, and
continuity. We then define a (pseudo) entanglement measure based on this
quantity, which we call the geometric squashed entanglement. We prove that the
geometric squashed entanglement is a 1-LOCC monotone (i.e., monotone
non-increasing with respect to local operations and classical communication
from Bob to Alice), that it vanishes if and only if the state on which it is
evaluated is unentangled, and that it reduces to the geometric measure of
entanglement if the state is pure. We also show that it is invariant with
respect to local isometries, subadditive, continuous, and normalized on
maximally entangled states. We next define the surprisal of measurement
recoverability, which is an information quantity in the spirit of quantum
discord, characterizing how well one can recover a share of a bipartite state
if it is measured. We prove that this discord-like quantity satisfies several
properties, including non-negativity, faithfulness on classical-quantum
states, invariance with respect to local isometries, a dimension bound, and
normalization on maximally entangled states. This quantity combined with a
recent breakthrough of Fawzi and Renner allows to characterize states with
discord nearly equal to zero as being approximate fixed points of entanglement
breaking channels (equivalently, they are recoverable from the state of a
measuring apparatus). Finally, we discuss a multipartite fidelity of recovery
and several of its properties.

\end{abstract}

\section{Introduction}

The conditional quantum mutual information (CQMI) is a central information
quantity that finds numerous applications in quantum information theory
\cite{DY08,YD09}, the theory of quantum correlations \cite{zurek,CW04}, and
quantum many-body physics \cite{K13thesis,B12}. For a quantum state
$\rho_{ABC}$ shared between three parties, say, Alice, Bob, and Charlie, the
CQMI is defined as
\begin{equation}
I(A;B|C)_{\rho}\equiv H(AC)_{\rho}+H(BC)_{\rho}-H(C)_{\rho}-H(ABC)_{\rho},
\label{eq:CMI-definition}%
\end{equation}
where $H(F)_{\sigma}\equiv-$Tr$\{\sigma_{F}\log\sigma_{F}\}$ is the von
Neumann entropy of a state $\sigma_{F}$ on system $F$ and we unambiguously let
$\rho_{C}\equiv\ $Tr$_{AB}\{\rho_{ABC}\}$ denote the reduced density operator
on system $C$, for example. The CQMI captures the correlations present between
Alice and Bob from the perspective of Charlie in the independent and
identically distributed (i.i.d.)~resource limit, where an asymptotically large
number of copies of the state $\rho_{ABC}$ are shared between the three
parties. It is non-negative \cite{PhysRevLett.30.434,LR73}, non-increasing
with respect to the action of local quantum operations on systems $A$ or $B$,
and obeys a duality relation for a four-party pure state $\psi_{ABCD}$, given
by $I(A;B|C)_{\psi}=I(B;A|D)_{\psi}$. It finds operational meaning as twice
the optimal quantum communication cost in the state redistribution
protocol~\cite{DY08,YD09}. It underlies the squashed entanglement~\cite{CW04},
which is a measure of entanglement that satisfies all of the axioms desired
for such a measure \cite{AF04,KW04,BCY11}, and furthermore underlies the
quantum discord \cite{zurek}, which is a measure of quantum correlations
different from those due to entanglement.

In an attempt to develop a version of the CQMI, which could potentially be
relevant for the \textquotedblleft one-shot\textquotedblright\ or finite
resource regimes, we along with Berta \cite{BSW14} recently proposed R\'{e}nyi
generalizations of the CQMI. We proved that these R\'{e}nyi generalizations of
the CQMI retain many of the properties of the original CQMI in
(\ref{eq:CMI-definition}). While the application of these particular R\'{e}nyi
CQMIs in one-shot state redistribution remains to be studied, (however, see
the recent progress on one-shot state redistribution in \cite{BCT14,DHO14}) we
have used them to define a R\'{e}nyi squashed entanglement and a R\'{e}nyi
quantum discord~\cite{SBW14}, which retain several properties of the
respective, original, von Neumann entropy based quantities.

One contribution of \cite{BSW14} was the conjecture that the proposed
R\'{e}nyi CQMIs are monotone increasing in the R\'{e}nyi parameter, as is
known to be the case for other R\'{e}nyi entropic quantities. That is, for a
tripartite state $\rho_{ABC}$, and for a R\'{e}nyi conditional mutual
information $\widetilde{I}_{\alpha}( A;B|C) _{\rho}$ defined as
\cite[Section~6]{BSW14}%
\begin{equation}
\widetilde{I}_{\alpha}( A;B|C) _{\rho}\equiv
\frac{1}{\alpha-1}\log\left\Vert \rho_{ABC}^{1/2}\rho_{AC}^{(1-\alpha
)/2\alpha}\rho_{C}^{(\alpha-1)/2\alpha}\rho_{BC}^{(1-\alpha)/2\alpha
}\right\Vert _{2\alpha}^{2\alpha},
\end{equation}
\cite[Section 8]{BSW14} conjectured that the following inequality holds for
$0\leq\alpha\leq\beta$:%
\begin{equation}
\widetilde{I}_{\alpha}( A;B|C) _{\rho}\leq\widetilde{I}_{\beta}( A;B|C)
_{\rho}. \label{eq:mono-alpha-2}%
\end{equation}
Proofs were given for this conjectured inequality when the R\'{e}nyi parameter
$\alpha$ is in a neighborhood of one and when $1/\alpha+1/\beta=2$
\cite[Section 8]{BSW14}.

We also pointed out implications of the conjectured inequality for
understanding states with small conditional quantum mutual information
\cite[Section 8]{BSW14}\ (later stressed in \cite{B14sem}). In particular, we
pointed out that the following lower bound on the conditional quantum mutual
information holds as a consequence of the conjectured inequality in
(\ref{eq:mono-alpha-2}) by choosing $\alpha=1/2$ and $\beta=1$:%
\begin{align}
I( A;B|C) _{\rho}  &  \geq-\log F\left(  \rho_{ABC},\mathcal{R}_{C\rightarrow
AC}^{P}\left(  \rho_{BC}\right)  \right) \label{pmap1}\\
&  \geq\frac{1}{4}\left\Vert \rho_{ABC}-\mathcal{R}_{C\rightarrow AC}%
^{P}\left(  \rho_{BC}\right)  \right\Vert _{1}^{2},
\end{align}
where $\mathcal{R}_{C\rightarrow AC}^{P}$ is a quantum channel known as the
Petz recovery map~\cite{Petz1986,Petz1988,Petz03,HJPW04}, defined as%
\begin{equation}
\mathcal{R}_{C\rightarrow AC}^{P}(\cdot)\equiv\rho_{AC}^{1/2}\rho_{C}%
^{-1/2}(\cdot)\rho_{C}^{-1/2}\rho_{AC}^{1/2}.
\end{equation}
The fidelity is a measure of how close two quantum states are and is defined
for positive semidefinite operators $P$ and $Q$ as%
\begin{equation}
F\left(  P,Q\right)  \equiv\left\Vert \sqrt{P}\sqrt{Q}\right\Vert _{1}^{2}.
\end{equation}
Throughout we denote the root fidelity by $\sqrt{F}\left(  P,Q\right)
\equiv\Vert\sqrt{P}\sqrt{Q}\Vert_{1}$. The trace distance bound in
(\ref{pmap1}) was conjectured previously in \cite{K13conj} and a related
conjecture (with a different lower bound) was considered in \cite{Winterconj}.

The conjectured inequality in (\ref{pmap1})\ revealed that (if it is true) it
would be possible to understand tripartite states with small conditional
mutual information in the following sense: \textit{If one loses system }%
$A$\textit{ of a tripartite state }$\rho_{ABC}$\textit{ and is allowed to
perform the Petz recovery map on system }$C$\textit{ alone, then the fidelity
of recovery in doing so will be high.} The converse statement was already
established in \cite[Proposition~35]{BSW14} and independently in
\cite[Eq.~(8)]{FR14}. Indeed, suppose now that a tripartite state $\rho_{ABC}$
has large conditional mutual information. Then if one loses system $A$ and
attempts to recover it by acting on system $C$ alone, then the fidelity of
recovery will not be high no matter what scheme is employed (see
\cite[Proposition~35]{BSW14} for specific parameters). These statements are
already known to be true for a classical system $C$, but the main question is
whether the inequality in (\ref{pmap1}) holds for a quantum system $C$.

\section{Summary of results}

\label{sec:summary} When studying the conjectured inequality in (\ref{pmap1}),
we can observe that a simple lower bound on the RHS\ is in terms of a quantity
that we call the \textit{surprisal of the fidelity of recovery}:%
\begin{align}
-\log F\left(  \rho_{ABC},\mathcal{R}_{C\rightarrow AC}^{P}\left(  \rho
_{BC}\right)  \right)   &  \geq I_{F}( A;B|C) _{\rho}\label{eq:I_F_defintion}%
\\
&  \equiv-\log F( A;B|C) _{\rho},
\end{align}
where the \textit{fidelity of recovery} is defined as%
\begin{equation}
F( A;B|C) _{\rho}\equiv\sup_{\mathcal{R}}F\left(  \rho_{ABC},\mathcal{R}%
_{C\rightarrow AC}\left(  \rho_{BC}\right)  \right)  . \label{eq:FoR}%
\end{equation}
That is, rather than considering the particular Petz recovery map, one could
consider optimizing the fidelity with respect to all such recovery maps. One
of the main objectives of the present paper is to study the fidelity of
recovery in more detail.

\textbf{Note:} After the completion of this work, we learned of the recent
breakthrough result of \cite{FR14}, in which the inequality $I(A;B|C)_{\rho
}\geq-\log F(A;B|C)_{\rho}$ was established for any tripartite state
$\rho_{ABC}\in\mathcal{S}(\mathcal{H}_{A}\otimes\mathcal{H}_{B}\otimes
\mathcal{H}_{C})$. Thus, for states with small conditional mutual information
(near to zero), the fidelity of recovery is high (near to one). Note that our
arXiv posting of the present work (arXiv:1410.1441) appeared one day after the
arXiv posting of \cite{FR14}. Furthermore note that the main result of
\cite{FR14} is now an easy corollary of the more general result in \cite{W15}.

\subsection{Properties of the surprisal of the fidelity of recovery}

Our conclusions for $I_{F}( A;B|C) _{\rho}$ are that it obeys many of the same
properties as the conditional mutual information $I\left(  A;B|C\right)
_{\rho}$:

\begin{enumerate}
\item \textbf{(Non-negativity)} $I_{F}( A;B|C) _{\rho}\geq0$ for any
tripartite quantum state, and for finite-dimensional $\rho_{ABC}$, $I_{F}(
A;B|C) _{\rho}=0$ if and only if $\rho_{ABC}$ is a short quantum Markov chain,
as defined in \cite{HJPW04}. A short quantum Markov chain is a tripartite
state $\rho_{ABC}$ for which $I( A;B|C) _{\rho}=0$, and such a state
necessarily has a particular structure, as elucidated in \cite{HJPW04}.

\item \textbf{(Monotonicity)} $I_{F}( A;B|C) _{\rho}$ is monotone with respect
to quantum operations on systems $A$ or $B$, in the sense that%
\begin{equation}
I_{F}( A;B|C) _{\rho}\geq I_{F}\left(  A^{\prime};B^{\prime}|C\right)
_{\omega},
\end{equation}
where $\omega_{ABC}\equiv\left(  \mathcal{N}_{A\rightarrow A^{\prime}}%
\otimes\mathcal{M}_{B\rightarrow B^{\prime}}\right)  \left(  \rho
_{ABC}\right)  $ and $\mathcal{N}_{A\rightarrow A^{\prime}}$ and
$\mathcal{M}_{B\rightarrow B^{\prime}}$ are quantum channels acting on systems
$A$ and $B$, respectively.

\item \textbf{(Local isometric invariance)} $I_{F}( A;B|C) _{\rho}$ is
invariant with respect to local isometries, in the sense that%
\begin{equation}
I_{F}( A;B|C) _{\rho}=I_{F}\left(  A^{\prime};B^{\prime}|C^{\prime}\right)
_{\sigma},
\end{equation}
where%
\begin{equation}
\sigma_{A^{\prime}B^{\prime}C^{\prime}}\equiv\left(  \mathcal{U}_{A\rightarrow
A^{\prime}}\otimes\mathcal{V}_{B\rightarrow B^{\prime}}\otimes\mathcal{W}%
_{C\rightarrow C^{\prime}}\right)  ( \rho_{ABC})
\end{equation}
and $\mathcal{U}_{A\rightarrow A^{\prime}}$, $\mathcal{V}_{B\rightarrow
B^{\prime}}$, and $\mathcal{W}_{C\rightarrow C^{\prime}}$ are isometric
quantum channels. An isometric channel $\mathcal{U}_{A\rightarrow A^{\prime}}$
has the following action on an operator $X_{A}$:
\begin{equation}
\mathcal{U}_{A\rightarrow A^{\prime}}(X_{A}) = U_{A \to A^{\prime}} X_{A} U_{A
\to A^{\prime}}^{\dag},
\end{equation}
where $U_{A \to A^{\prime}}$ is an isometry, satisfying $U_{A \to A^{\prime}%
}^{\dag}U_{A \to A^{\prime}} = I_{A}$.

\item \textbf{(Duality)}\ For a four-party pure state $\psi_{ABCD}$, the
following duality relation holds%
\begin{equation}
I_{F}( A;B|C) _{\psi}=I_{F}( A;B|D) _{\psi}.
\end{equation}

\item \textbf{(Dimension bound)} The following dimension bound holds%
\begin{equation}
I_{F}( A;B|C) _{\rho}\leq2\log\left\vert A\right\vert ,
\end{equation}
where $\left\vert A\right\vert $ is the dimension of the system $A$. If the
system $A$ is classical, so that we relabel it as $X$, then%
\begin{equation}
I_{F}( X;B|C) _{\rho}\leq\log\left\vert X\right\vert .
\end{equation}
By a classical system $X$, we mean that $\rho_{XBC}$ has the following form:
\begin{equation}
\rho_{XBC} = \sum_{x} p_{X}(x) \vert x \rangle\langle x \vert\otimes\rho
^{x}_{BC},
\end{equation}
for some probability distribution $p_{X}(x)$, orthonormal basis $\{ \vert x
\rangle\}$, and set $\{\rho^{x}_{BC}\}$ of density operators.

\item \textbf{(Continuity)} If two quantum states $\rho_{ABC}$ and
$\sigma_{ABC}$ are close to each other in the sense that $F\left(  \rho
_{ABC},\sigma_{ABC}\right)  \approx1$, then $I_{F}( A;B|C) _{\rho}\approx
I_{F}( A;B|C) _{\sigma}$.

\item \textbf{(Weak chain rule)} The chain rule for conditional mutual
information of a four-party state $\rho_{ABCD}$ is as follows:%
\begin{equation}
I\left(  AC;B|D\right)  _{\rho}=I\left(  A;B|CD\right)  _{\rho}+I\left(
C;B|D\right)  _{\rho}.
\end{equation}
We find something weaker than this for $I_{F}$, which we call the weak chain
rule for $I_{F}$:%
\begin{equation}
I_{F}\left(  AC;B|D\right)  _{\rho}\geq I_{F}\left(  A;B|CD\right)  _{\rho}.
\end{equation}

\end{enumerate}

Let us note here that, by inspecting the definitions, the fidelity of recovery
$F( A;B|C) _{\rho}$ and $I_{F}( A;B|C) _{\rho}$ are clearly not symmetric
under the exchange of the $A$ and $B$ systems, unlike the conditional mutual
information $I( A;B|C) _{\rho}$. Thus we might also refer to $I_{F}( A;B|C)
_{\rho}$ as the conditional information that $B$ has about $A$ from the
perspective of $C$.

\subsection{Geometric squashed entanglement}

Our next contribution is to define a (pseudo) entanglement measure of a
bipartite state that we call the \textit{geometric squashed entanglement}. To
motivate this quantity, recall that the squashed entanglement of a bipartite
state $\rho_{AB}$ is defined as%
\begin{equation}
E^{\operatorname{sq}}( A;B) _{\rho}\equiv
\frac{1}{2}\inf_{\omega_{ABE}}\left\{  I( A;B|E) _{\omega}:\rho_{AB}%
=\operatorname{Tr}_{E}\left\{  \omega_{ABE}\right\}  \right\}  ,
\end{equation}
where the infimum is over all extensions $\omega_{ABE}$ of the state
$\rho_{AB}$ \cite{CW04}. The interpretation of $E^{\operatorname{sq}}\left(
A;B\right)  _{\rho}$\ is that it quantifies the correlations present between
Alice and Bob after a third party (often associated to an environment or
eavesdropper) attempts to \textquotedblleft squash down\textquotedblright%
\ their correlations. In light of the above discussion, we define the
geometric squashed entanglement simply by replacing the conditional mutual
information with $I_{F}$:%
\begin{equation}
E_{F}^{\operatorname{sq}}( A;B) _{\rho}\equiv
\frac{1}{2}\inf_{\omega_{ABE}}\left\{  I_{F}( A;B|E) _{\omega}:\rho
_{AB}=\operatorname{Tr}_{E}\left\{  \omega_{ABE}\right\}  \right\}  .
\end{equation}
We also employ the related quantity throughout the paper:%
\begin{equation}
F^{\operatorname{sq}}( A;B) _{\rho}\equiv\sup_{\omega_{ABE}}\left\{  F( A;B|E)
_{\rho}:\rho_{AB}=\operatorname{Tr}_{E}\left\{  \omega_{ABE}\right\}
\right\}  ,
\end{equation}
with the two of them being related by%
\begin{equation}
E_{F}^{\operatorname{sq}}( A;B) _{\rho}=-\frac{1}{2}\log F^{\operatorname{sq}%
}( A;B) _{\rho}.
\end{equation}

We prove the following results for the geometric squashed entanglement:

\begin{enumerate}
\item \textbf{(1-LOCC Monotone)} The geometric squashed entanglement of
$\rho_{AB}$\ does not increase with respect to local operations and classical
communication from Bob to Alice. That is, the following inequality holds%
\begin{equation}
E_{F}^{\operatorname{sq}}( A;B) _{\rho}\geq E_{F}^{\operatorname{sq}}\left(
A^{\prime};B^{\prime}\right)  _{\omega},
\end{equation}
where $\omega_{AB}\equiv\Lambda_{AB\rightarrow A^{\prime}B^{\prime}}\left(
\rho_{AB}\right)  $ and $\Lambda_{AB\rightarrow A^{\prime}B^{\prime}}$ is a
quantum channel realized by local operations and classical communication from
Bob to Alice. (Due to the asymmetric nature of the fidelity of recovery, we do
not seem to be able to prove that the geometric squashed entanglement is an
LOCC\ monotone.) The geometric squashed entanglement is also convex, i.e.,%
\begin{equation}
\sum_{x}p_{X}( x) E_{F}^{\operatorname{sq}}( A;B) _{\rho^{x}}\geq
E_{F}^{\operatorname{sq}}( A;B) _{\overline{\rho}},
\end{equation}
where%
\begin{equation}
\overline{\rho}_{AB}\equiv\sum_{x}p_{X}( x) \rho_{AB}^{x},
\end{equation}
$p_{X}$ is a probability distribution and $\left\{  \rho_{AB}^{x}\right\}  $
is a set of states.

\item (\textbf{Local isometric invariance}) $E_{F}^{\operatorname{sq}}\left(
A;B\right)  _{\rho}$ is invariant with respect to local isometries, in the
sense that%
\begin{equation}
E_{F}^{\operatorname{sq}}( A;B) _{\rho}=E_{F}^{\operatorname{sq}}(A^{\prime
};B^{\prime})_{\sigma},
\end{equation}
where%
\begin{equation}
\sigma_{A^{\prime}B^{\prime}}\equiv\left(  \mathcal{U}_{A\rightarrow
A^{\prime}}\otimes\mathcal{V}_{B\rightarrow B^{\prime}}\right)  \left(
\rho_{AB}\right)
\end{equation}
and $\mathcal{U}_{A\rightarrow A^{\prime}}$ and $\mathcal{V}_{B\rightarrow
B^{\prime}}$ are isometric quantum channels.

\item \textbf{(Faithfulness)}\ The geometric squashed entanglement of
$\rho_{AB}$\ is equal to zero if and only if $\rho_{AB}$ is a separable
(unentangled) state. In particular, we prove the following bound by appealing
directly to the argument in \cite{Winterconj}:%
\begin{equation}
E_{F}^{\operatorname{sq}}( A;B) _{\rho}\geq\frac{1}{512\left\vert A\right\vert
^{4}}\left\Vert \rho_{AB}-\text{SEP}(A:B)\right\Vert _{1}^{4},
\end{equation}
where the trace distance to separable states is defined by%
\begin{equation}
\left\Vert \rho_{AB}-\text{SEP}(A:B)\right\Vert _{1}\equiv
\inf_{\sigma_{AB}\in\text{SEP}\left(  A:B\right)  }\left\Vert \rho_{AB}%
-\sigma_{AB}\right\Vert _{1}.
\end{equation}

\item \textbf{(Reduction to geometric measure)} The geometric squashed
entanglement of a pure state $\left\vert \phi\right\rangle _{AB}$\ reduces to
the well known geometric measure of entanglement \cite{WG03} (see also
\cite{CAH13}\ and references therein):%
\begin{align}
E_{F}^{\operatorname{sq}}( A;B) _{\psi}  &  =-\frac{1}{2}\log\sup_{\left\vert
\varphi\right\rangle _{A}}\left\langle \phi\right\vert _{AB}\left(
\varphi_{A}\otimes\phi_{B}\right)  \left\vert \phi\right\rangle _{AB}\\
&  =-\log\left\Vert \phi_{A}\right\Vert _{\infty}.
\end{align}
Recall that the geometric measure of $\left\vert \phi\right\rangle _{AB}$\ is
known to be equal to%
\begin{equation}
-\log\sup_{\left\vert \varphi\right\rangle _{A},\left\vert \psi\right\rangle
_{B}}\left\langle \phi\right\vert _{AB}\left(  \varphi_{A}\otimes\psi
_{B}\right)  \left\vert \phi\right\rangle _{AB}=
-\log\left\Vert \phi_{A}\right\Vert _{\infty},
\end{equation}
where $\left\Vert A\right\Vert _{\infty}$ is the infinity norm of an operator
$A$, equal to its largest singular value. (Note that the above quantity is
often referred to as the \textit{logarithmic geometric measure of
entanglement}. Here, for brevity, we simply refer to it as the geometric measure.)

\item \textbf{(Normalization)} The geometric squashed entanglement of a
maximally entangled state $\Phi_{AB}$ is equal to $\log d$, where $d$ is the
Schmidt rank of $\Phi_{AB}$. It is larger than $\log d$ when evaluated for a
private state \cite{HHHO05,HHHO09}\ of $\log d$ private bits.

\item \textbf{(Subadditivity)} The geometric squashed entanglement is
subadditive for tensor-product states, i.e.,%
\begin{equation}
E_{F}^{\operatorname{sq}}\left(  A_{1}A_{2};B_{1}B_{2}\right)  _{\omega}\leq
E_{F}^{\operatorname{sq}}\left(  A_{1};B_{1}\right)  _{\rho}+E_{F}%
^{\operatorname{sq}}\left(  A_{2};B_{2}\right)  _{\sigma},
\end{equation}
where $\omega_{A_{1}B_{1}A_{2}B_{2}}\equiv\rho_{A_{1}B_{1}}\otimes
\sigma_{A_{2}B_{2}}$.

\item \textbf{(Continuity)} If two quantum states $\rho_{AB}$ and $\sigma
_{AB}$ are close in trace distance, then their respective geometric squashed
entanglements are close as well.
\end{enumerate}

\subsection{Surprisal of measurement recoverability}

The quantum discord $D( \overline{A};B) _{\rho}$\ is an information quantity
which characterizes quantum correlations of a bipartite state $\rho_{AB}$, by
quantifying how much correlation is lost through the act of a quantum
measurement \cite{Z00,zurek}\ (we give a full definition later on). By a chain
of reasoning detailed in Section~\ref{sec:FoMR} which begins with the original
definition of quantum discord, we define the surprisal of measurement
recoverability of a bipartite state\ as follows:%
\begin{equation}
D_{F}( \overline{A};B) _{\rho}\equiv-\log\sup_{\mathcal{E}_{A}}F( \rho
_{AB},\mathcal{E}_{A}( \rho_{AB}) ) ,
\end{equation}
where the supremum is over the convex set of entanglement breaking channels
\cite{HSR03}. Since every entanglement breaking channel can be written as a
concatenation of a measurement map followed by a preparation map, $D_{F}(
\overline{A};B) _{\rho}$ characterizes how well one can recover a bipartite
state after performing a quantum measurement on one share of it. Equivalently,
the quantity captures how close $\rho_{AB}$ is to being a fixed point of an
entanglement breaking channel.

We establish several properties of $D_{F}( \overline{A};B) _{\rho}$, which are
analogous to properties known to hold for the quantum discord \cite{KBCPV12}:

\begin{enumerate}
\item \textbf{(Non-negativity)} This follows trivially because the fidelity
between two quantum states is always a real number between zero and one.

\item \textbf{(Local isometric invariance)} $D_{F}\left(  \overline
{A};B\right)  _{\rho}$ is invariant with respect to local isometries, in the
sense that%
\begin{equation}
D_{F}( \overline{A};B) _{\rho}=D_{F}(\overline{A^{\prime}};B^{\prime}%
)_{\sigma},
\end{equation}
where%
\begin{equation}
\sigma_{A^{\prime}B^{\prime}}\equiv\left(  \mathcal{U}_{A\rightarrow
A^{\prime}}\otimes\mathcal{V}_{B\rightarrow B^{\prime}}\right)  \left(
\rho_{AB}\right)
\end{equation}
and $\mathcal{U}_{A\rightarrow A^{\prime}}$ and $\mathcal{V}_{B\rightarrow
B^{\prime}}$ are isometric quantum channels.

\item \textbf{(Faithfulness)}\ $D_{F}( \overline{A};B) _{\rho}$ is equal to
zero if and only if $\rho_{AB}$ is a classical-quantum state (classical on
system $A$).

\item \textbf{(Dimension bound)} $D_{F}( \overline{A};B) _{\rho}\leq
\log\left\vert A\right\vert $.

\item \textbf{(Normalization)} $D_{F}( \overline{A};B) _{\Phi}$ for a
maximally entangled state $\Phi_{AB}$ is equal to $\log d$, where $d$ is the
Schmidt rank of $\Phi_{AB}$.

\item \textbf{(Monotonicity)} The surprisal of measurement recoverability is
monotone with respect to quantum operations on the unmeasured system, i.e.,%
\begin{equation}
D_{F}( \overline{A};B) _{\rho}\geq D_{F}\left(  \overline{A};B^{\prime
}\right)  _{\sigma},
\end{equation}
where $\sigma_{AB^{\prime}}\equiv\mathcal{N}_{B\rightarrow B^{\prime}}\left(
\rho_{AB}\right)  $.

\item \textbf{(Continuity)} If two quantum states $\rho_{AB}$ and $\sigma
_{AB}$ are close in trace distance, then the respective $D_{F}\left(
\overline{A};B\right)  $ quantities\ are close as well.
\end{enumerate}

Finally, we use $D_{F}( \overline{A};B) _{\rho}$ and a recent result of Fawzi
and Renner \cite{FR14}\ to establish that the quantum discord of $\rho_{AB}$
is nearly equal to zero if and only if $\rho_{AB}$ is an approximate fixed
point of entanglement breaking channel (i.e., if it is possible to nearly
recover $\rho_{AB}$ after performing a measurement on the system $A$). We then
argue that several discord-like measures appearing throughout the literature
\cite{KBCPV12}\ have a more natural physical grounding if they are based on
how far a given bipartite state is from being a fixed point of an entanglement
breaking channel.

\section{Preliminaries}

\label{sec:notation}\textbf{Norms, states, extensions, channels, and
measurements.} Let $\mathcal{B}\left(  \mathcal{H}\right)  $ denote the
algebra of bounded linear operators acting on a Hilbert space $\mathcal{H}$.
We restrict ourselves to finite-dimensional Hilbert spaces throughout this
paper. For $\alpha\geq1$, we define the $\alpha$-norm of an operator $X$ as
\begin{equation}
\left\Vert X\right\Vert _{\alpha}\equiv\operatorname{Tr}\{(\sqrt{X^{\dag}%
X})^{\alpha}\}^{1/\alpha} . \label{eq:a-norm}%
\end{equation}
Let $\mathcal{B}\left(  \mathcal{H}\right)  _{+}$ denote the subset of
positive semi-definite operators. We also write $X\geq0$ if $X\in
\mathcal{B}\left(  \mathcal{H}\right)  _{+}$. An\ operator $\rho$ is in the
set $\mathcal{S}\left(  \mathcal{H}\right)  $\ of density operators (or
states) if $\rho\in\mathcal{B}\left(  \mathcal{H}\right)  _{+}$ and
Tr$\left\{  \rho\right\}  =1$. The tensor product of two Hilbert spaces
$\mathcal{H}_{A}$ and $\mathcal{H}_{B}$ is denoted by $\mathcal{H}_{A}%
\otimes\mathcal{H}_{B}$ or $\mathcal{H}_{AB}$.\ Given a multipartite density
operator $\rho_{AB}\in\mathcal{S}(\mathcal{H}_{A}\otimes\mathcal{H}_{B})$, we
unambiguously write $\rho_{A}=\ $Tr$_{B}\left\{  \rho_{AB}\right\}  $ for the
reduced density operator on system $A$. We use $\rho_{AB}$, $\sigma_{AB}$,
$\tau_{AB}$, $\omega_{AB}$, etc.~to denote general density operators in
$\mathcal{S}(\mathcal{H}_{A}\otimes\mathcal{H}_{B})$, while $\psi_{AB}$,
$\varphi_{AB}$, $\phi_{AB}$, etc.~denote rank-one density operators (pure
states) in $\mathcal{S}(\mathcal{H}_{A}\otimes\mathcal{H}_{B})$ (with it
implicit, clear from the context, and the above convention implying that
$\psi_{A}$, $\varphi_{A}$, $\phi_{A}$ are mixed if $\psi_{AB}$, $\varphi_{AB}%
$, $\phi_{AB}$ are pure and entangled).

We also say that pure-state vectors $|\psi\rangle$ in $\mathcal{H}$ are
states. Any bipartite pure state $|\psi\rangle_{AB}$ in $\mathcal{H}_{AB}$ is
written in Schmidt form as
\begin{equation}
\left\vert \psi\right\rangle _{AB}\equiv\sum_{i=0}^{d-1}\sqrt{\lambda_{i}%
}\left\vert i\right\rangle _{A}\left\vert i\right\rangle _{B},
\end{equation}
where $\{|i\rangle_{A}\}$ and $\{|i\rangle_{B}\}$ form orthonormal bases in
$\mathcal{H}_{A}$ and $\mathcal{H}_{B}$, respectively, $\lambda_{i}>0$ for all
$i$, $\sum_{i=0}^{d-1}\lambda_{i}=1$, and $d$ is the Schmidt rank of the
state. By a maximally entangled state, we mean a bipartite pure state of the
form
\begin{equation}
\left\vert \Phi\right\rangle _{AB}\equiv\frac{1}{\sqrt{d}}\sum_{i=0}%
^{d-1}\left\vert i\right\rangle _{A}\left\vert i\right\rangle _{B}.
\end{equation}

A state $\gamma_{ABA^{\prime}B^{\prime}}$\ is a private state
\cite{HHHO05,HHHO09} if Alice and Bob can extract a secret key from it by
performing local von Neumann measurements on the $A$ and $B$ systems of
$\gamma_{ABA^{\prime}B^{\prime}}$, such that the resulting secret key is
product with any purifying system of $\gamma_{ABA^{\prime}B^{\prime}}$. The
systems $A^{\prime}$ and $B^{\prime}$ are known as \textquotedblleft shield
systems\textquotedblright\ because they aid in keeping the key secure from any
eavesdropper possessing the purifying system. Interestingly, a private state
of $\log d$ private bits can be written in the following form
\cite{HHHO05,HHHO09}:%
\begin{equation}
\gamma_{ABA^{\prime}B^{\prime}}=U_{ABA^{\prime}B^{\prime}}\left(  \Phi
_{AB}\otimes\rho_{A^{\prime}B^{\prime}}\right)  U_{ABA^{\prime}B^{\prime}%
}^{\dag}, \label{eq:private-1}%
\end{equation}
where%
\begin{equation}
U_{ABA^{\prime}B^{\prime}}=\sum_{i,j}\left\vert i\right\rangle \left\langle
i\right\vert _{A}\otimes\left\vert j\right\rangle \left\langle j\right\vert
_{B}\otimes U_{A^{\prime}B^{\prime}}^{ij}.
\end{equation}
The unitaries can be chosen such that $U_{A^{\prime}B^{\prime}}^{ij}%
=V_{A^{\prime}B^{\prime}}^{j}$ or $U_{A^{\prime}B^{\prime}}^{ij}=V_{A^{\prime
}B^{\prime}}^{i}$. This implies that the unitary $U_{ABA^{\prime}B^{\prime}}$
can be implemented either as%
\begin{equation}
U_{ABA^{\prime}B^{\prime}}=\sum_{i}\left\vert i\right\rangle \left\langle
i\right\vert _{A}\otimes I_{B}\otimes V_{A^{\prime}B^{\prime}}^{i}%
\end{equation}
or%
\begin{equation}
U_{ABA^{\prime}B^{\prime}}=I_{A}\otimes\sum_{i}\left\vert i\right\rangle
\left\langle i\right\vert _{B}\otimes V_{A^{\prime}B^{\prime}}^{i}.
\label{eq:private-last}%
\end{equation}

The trace distance between two quantum states $\rho,\sigma\in\mathcal{S}%
\left(  \mathcal{H}\right)  $\ is equal to $\left\Vert \rho-\sigma\right\Vert
_{1}$. It has a direct operational interpretation in terms of the
distinguishability of these states. That is, if $\rho$ or $\sigma$ is prepared
with equal probability and the task is to distinguish them via some quantum
measurement, then the optimal success probability in doing so is equal to
$\left(  1+\left\Vert \rho-\sigma\right\Vert _{1}/2\right)  /2$.

A linear map $\mathcal{N}_{A\rightarrow B}:\mathcal{B}\left(  \mathcal{H}%
_{A}\right)  \rightarrow\mathcal{B}\left(  \mathcal{H}_{B}\right)  $\ is
positive if $\mathcal{N}_{A\rightarrow B}\left(  \sigma_{A}\right)
\in\mathcal{B}\left(  \mathcal{H}_{B}\right)  _{+}$ whenever $\sigma_{A}%
\in\mathcal{B}\left(  \mathcal{H}_{A}\right)  _{+}$. Let id$_{A}$ denote the
identity map acting on a system $A$. A linear map $\mathcal{N}_{A\rightarrow
B}$ is completely positive if the map id$_{R}\otimes\mathcal{N}_{A\rightarrow
B}$ is positive for a reference system $R$ of arbitrary size. A linear map
$\mathcal{N}_{A\rightarrow B}$ is trace-preserving if Tr$\left\{
\mathcal{N}_{A\rightarrow B}\left(  \tau_{A}\right)  \right\}  =\ $Tr$\left\{
\tau_{A}\right\}  $ for all input operators $\tau_{A}\in\mathcal{B}\left(
\mathcal{H}_{A}\right)  $. If a linear map is completely positive and
trace-preserving (CPTP), we say that it is a quantum channel or quantum
operation. An extension of a state $\rho_{A}\in\mathcal{S}\left(
\mathcal{H}_{A}\right)  $ is some state $\Omega_{RA}\in\mathcal{S}\left(
\mathcal{H}_{R}\otimes\mathcal{H}_{A}\right)  $ such that $\mathrm{Tr}%
_{R}\left\{  \Omega_{RA}\right\}  =\rho_{A}$. An isometric extension
$U_{A\rightarrow BE}^{\mathcal{N}}$ of a channel $\mathcal{N}_{A\rightarrow
B}$ acting on a state $\rho_{A}\in\mathcal{S}(\mathcal{H}_{A})$ is a linear
map that satisfies the following:
\begin{align}
\mathrm{Tr}_{E}\left\{  U_{A\rightarrow BE}^{\mathcal{N}}\rho_{A}%
(U_{A\rightarrow BE}^{\mathcal{N}})^{\dag}\right\}   &  =\mathcal{N}%
_{A\rightarrow B}\left(  \rho_{A}\right)  ,\\
U_{\mathcal{N}}^{\dagger}U_{\mathcal{N}}  &  =I_{A},\\
U_{\mathcal{N}}U_{\mathcal{N}}^{\dagger}  &  =\Pi_{BE},
\end{align}
where $\Pi_{BE}$ is a projection onto a subspace of the Hilbert space
$\mathcal{H}_{B}\otimes\mathcal{H}_{E}$.

\section{Fidelity of recovery}

\label{sec:fidelity-of-recovery}In this section, we formally define the
fidelity of recovery for a tripartite state $\rho_{ABC}$, and we prove that it
possesses various properties, demonstrating that the quantity $I_{F}\left(
A;B|C\right)  _{\rho}$ defined\ in (\ref{eq:I_F_defintion}) is similar to the
conditional mutual information.

\begin{definition}
[Fidelity of recovery]\label{def:FoR}Let $\rho_{ABC}$ be a tripartite state.
The fidelity of recovery for $\rho_{ABC}$ with respect to system $A$ is
defined as follows:%
\begin{equation}
F( A;B|C) _{\rho}\equiv\sup_{\mathcal{R}_{C\rightarrow AC}}F\left(  \rho
_{ABC},\mathcal{R}_{C\rightarrow AC}\left(  \rho_{BC}\right)  \right)  .
\end{equation}
This quantity characterizes how well one can recover the full state on systems
$ABC$ from system $C$ alone if system $A$ is lost.
\end{definition}

\begin{proposition}
[Non-negativity]Let $\rho_{ABC}$ be a tripartite state. Then $I_{F}\left(
A;B|C\right)  _{\rho}\geq0$, and for finite-dimensional $\rho_{ABC}$, $I_{F}(
A;B|C) _{\rho}=0$ if and only if $\rho_{ABC}$ is a short quantum Markov chain,
as defined in \cite{HJPW04}.
\end{proposition}

\begin{proof}
The inequality $I_{F}( A;B|C) _{\rho}\geq0$ is a consequence of the fidelity
always being less than or equal to one. Suppose that $\rho_{ABC}$ is a short
quantum Markov chain as defined in \cite{HJPW04}. As discussed in that paper,
this is equivalent to the equality%
\begin{equation}
\rho_{ABC}=\mathcal{R}_{C\rightarrow AC}^{P}\left(  \rho_{BC}\right)  ,
\end{equation}
where $\mathcal{R}_{C\rightarrow AC}^{P}$ is the Petz recovery channel. So
this implies that%
\begin{equation}
F\left(  \rho_{ABC},\mathcal{R}_{C\rightarrow AC}^{P}\left(  \rho_{BC}\right)
\right)  =1,
\end{equation}
which in turn implies that $F( A;B|C) _{\rho}=1$ and hence $I_{F}( A;B|C)
_{\rho}=0$. Now suppose that $I_{F}\left(  A;B|C\right)  _{\rho}=0$. This
implies that%
\begin{equation}
\sup_{\mathcal{R}_{C\rightarrow AC}}F\left(  \rho_{ABC},\mathcal{R}%
_{C\rightarrow AC}\left(  \rho_{BC}\right)  \right)  =1.
\end{equation}
Due to the finite-dimensional assumption, the space of channels over which we
are optimizing is compact. Furthermore, the fidelity is continuous in its
arguments. This is sufficient for us to conclude that the supremum is achieved
and that there exists a channel $\mathcal{R}_{C\rightarrow AC}$ for which
$F\left(  \rho_{ABC},\mathcal{R}_{C\rightarrow AC}\left(  \rho_{BC}\right)
\right)  =1$, implying that%
\begin{equation}
\rho_{ABC}=\mathcal{R}_{C\rightarrow AC}\left(  \rho_{BC}\right)  .
\end{equation}
From a result of Petz \cite{Petz1988}, this implies that the Petz recovery
channel recovers $\rho_{ABC}$ perfectly, i.e.,%
\begin{equation}
\rho_{ABC}=\mathcal{R}_{C\rightarrow AC}^{P}\left(  \rho_{BC}\right)  ,
\end{equation}
and this is equivalent to $\rho_{ABC}$ being a short quantum Markov chain
\cite{HJPW04}.
\end{proof}

\begin{proposition}
[Monotonicity]\label{prop:FoR-mono-local-ops}The fidelity of recovery is
monotone with respect to local operations on systems $A$ and $B$, in the sense
that%
\begin{equation}
F( A;B|C) _{\rho}\leq F\left(  A^{\prime};B^{\prime}|C\right)  _{\tau},
\end{equation}
where $\tau_{A^{\prime}B^{\prime}C}\equiv\left(  \mathcal{N}_{A\rightarrow
A^{\prime}}\otimes\mathcal{M}_{B\rightarrow B^{\prime}}\right)  \left(
\rho_{ABC}\right)  $. The above inequality is equivalent to%
\begin{equation}
I_{F}( A;B|C) _{\rho}\geq I_{F}\left(  A^{\prime};B^{\prime}|C\right)  _{\tau
}.
\end{equation}

\end{proposition}

\begin{proof}
For any recovery map $\mathcal{R}_{C\rightarrow AC}$, we have
that
\begin{align}
& F\left(  \rho_{ABC},\mathcal{R}_{C\rightarrow AC}\left(  \rho_{BC}\right)
\right)  \nonumber \\
& \leq F\left(  \left(  \mathcal{N}_{A\rightarrow A^{\prime}}\otimes
\mathcal{M}_{B\rightarrow B^{\prime}}\right)  (  \rho_{ABC})
,\left(  \mathcal{N}_{A\rightarrow A^{\prime}}\otimes\mathcal{M}_{B\rightarrow
B^{\prime}}\right)  \left(  \mathcal{R}_{C\rightarrow AC}\left(  \rho
_{BC}\right)  \right)  \right)  \\
&  =F\left(  \left(  \mathcal{N}_{A\rightarrow A^{\prime}}\otimes
\mathcal{M}_{B\rightarrow B^{\prime}}\right)  (  \rho_{ABC})
,\left(  \mathcal{N}_{A\rightarrow A^{\prime}}\circ\mathcal{R}_{C\rightarrow
AC}\right)  \left(  \mathcal{M}_{B\rightarrow B^{\prime}}\left(  \rho
_{BC}\right)  \right)  \right)  \\
&  \leq\sup_{\mathcal{R}_{C\rightarrow A^{\prime}C}}F\left(  \left(
\mathcal{N}_{A\rightarrow A^{\prime}}\otimes\mathcal{M}_{B\rightarrow
B^{\prime}}\right)  (  \rho_{ABC})  ,\mathcal{R}_{C\rightarrow
A^{\prime}C}\left(  \mathcal{M}_{B\rightarrow B^{\prime}}\left(  \rho
_{BC}\right)  \right)  \right)  \\
&  =F\left(  A^{\prime};B^{\prime}|C\right)  _{\left(  \mathcal{N\otimes
M}\right)  \left(  \rho\right)  },
\end{align}
where the first inequality is due to monotonicity of the
fidelity with respect to quantum operations. Since the chain of inequalities
holds for all $\mathcal{R}_{C\rightarrow AC}$, it follows that%
\begin{align}
F( A;B|C) _{\rho}  &  =\sup_{\mathcal{R}_{C\rightarrow AC}}F\left(  \rho
_{ABC},\mathcal{R}_{C\rightarrow AC}\left(  \rho_{BC}\right)  \right)
\label{eq:ops-on-A}\\
&  \leq F\left(  A^{\prime};B^{\prime}|C\right)  _{\left(  \mathcal{N\otimes
M}\right)  \left(  \rho\right)  }.
\end{align}

\end{proof}

\begin{remark}
The physical interpretation of the above monotonicity with respect to local
operations is as follows:\ for a tripartite state $\rho_{ABC}$, suppose that
system $A$ is lost. Then it is easier to recover the state on systems $ABC$
from $C$ alone if there is local noise applied to systems $A$ or $B$ or both,
before system $A$ is lost (and thus before attempting the recovery).
\end{remark}

\begin{proposition}
[Local isometric invariance]\label{prop:FoR-local-iso}Let $\rho_{ABC}$ be a
tripartite quantum state and let%
\begin{equation}
\sigma_{A^{\prime}B^{\prime}C^{\prime}}\equiv\left(  \mathcal{U}_{A\rightarrow
A^{\prime}}\otimes\mathcal{V}_{B\rightarrow B^{\prime}}\otimes\mathcal{W}%
_{C\rightarrow C^{\prime}}\right)  ( \rho_{ABC}) ,
\end{equation}
where $\mathcal{U}_{A\rightarrow A^{\prime}}$, $\mathcal{V}_{B\rightarrow
B^{\prime}}$, and $\mathcal{W}_{C\rightarrow C^{\prime}}$ are isometric
quantum channels. Then%
\begin{align}
F( A;B|C) _{\rho}  &  =F\left(  A^{\prime};B^{\prime}|C^{\prime}\right)
_{\sigma},\\
I_{F}( A;B|C) _{\rho}  &  =I_{F}\left(  A^{\prime};B^{\prime}|C^{\prime
}\right)  _{\sigma}.
\end{align}

\end{proposition}

\begin{proof}
We prove the statement for fidelity of recovery. We first need to define some
CPTP\ maps that invert the isometric channels $\mathcal{U}_{A\rightarrow
A^{\prime}}$, $\mathcal{V}_{B\rightarrow B^{\prime}}$, and $\mathcal{W}%
_{C\rightarrow C^{\prime}}$, given that $\mathcal{U}_{A\rightarrow A^{\prime}%
}^{\dag}$, $\mathcal{V}_{B\rightarrow B^{\prime}}^{\dag}$, and $\mathcal{W}%
_{C\rightarrow C^{\prime}}^{\dag}$ are not necessarily quantum channels. So we
define the CPTP\ linear map $\mathcal{T}_{A^{\prime}\rightarrow A}%
^{\mathcal{U}}$ as follows:%
\begin{equation}
\mathcal{T}_{A^{\prime}\rightarrow A}^{\mathcal{U}}\left(  \omega_{A^{\prime}%
}\right)  \equiv\mathcal{U}_{A\rightarrow A^{\prime}}^{\dag}\left(
\omega_{A^{\prime}}\right) \label{eq:T-maps}
+\text{Tr}\left\{  \left(  \text{id}_{A^{\prime}}-\mathcal{U}_{A\rightarrow
A^{\prime}}^{\dag}\right)  \left(  \omega_{A^{\prime}}\right)  \right\}
\tau_{A},
\end{equation}
where $\tau_{A}$ is some state on system $A$. We define the maps
$\mathcal{T}_{B^{\prime}\rightarrow B}^{\mathcal{V}}$ and $\mathcal{T}%
_{C^{\prime}\rightarrow C}^{\mathcal{W}}$ similarly. All three maps have the
property that%
\begin{align}
\mathcal{T}_{A^{\prime}\rightarrow A}^{\mathcal{U}}\circ\mathcal{U}%
_{A\rightarrow A^{\prime}}  &  =\text{id}_{A},\label{eq:invert-isometry-A}\\
\mathcal{T}_{B^{\prime}\rightarrow B}^{\mathcal{V}}\circ\mathcal{V}%
_{B\rightarrow B^{\prime}}  &  =\text{id}_{B},\label{eq:invert-isometry-B}\\
\mathcal{T}_{C^{\prime}\rightarrow C}^{\mathcal{W}}\circ\mathcal{W}%
_{C\rightarrow C^{\prime}}  &  =\text{id}_{C}. \label{eq:invert-isometry-C}%
\end{align}
Let $\mathcal{R}_{C\rightarrow AC}$ be an arbitrary recovery map. Then

\begin{align}
&  F\left(  \rho_{ABC},\mathcal{R}_{C\rightarrow AC}\left(  \rho_{BC}\right)
\right)  \nonumber\\
&  =F\left(  \left(  \mathcal{U}_{A\rightarrow A^{\prime}}\otimes
\mathcal{V}_{B\rightarrow B^{\prime}}\otimes\mathcal{W}_{C\rightarrow
C^{\prime}}\right)  (  \rho_{ABC})  ,\left(  \mathcal{U}%
_{A\rightarrow A^{\prime}}\otimes\mathcal{V}_{B\rightarrow B^{\prime}}%
\otimes\mathcal{W}_{C\rightarrow C^{\prime}}\right)  \left(  \mathcal{R}%
_{C\rightarrow AC}\left(  \rho_{BC}\right)  \right)  \right)  \\
&  =F\left(  \sigma_{A^{\prime}B^{\prime}C^{\prime}},\left(  \mathcal{U}%
_{A\rightarrow A^{\prime}}\otimes\mathcal{W}_{C\rightarrow C^{\prime}}\right)
\left(  \mathcal{R}_{C\rightarrow AC}\left(  \mathcal{V}_{B\rightarrow
B^{\prime}}\left(  \rho_{BC}\right)  \right)  \right)  \right)  \\
&  =F\left(  \sigma_{A^{\prime}B^{\prime}C^{\prime}},\left(  \mathcal{U}%
_{A\rightarrow A^{\prime}}\otimes\mathcal{W}_{C\rightarrow C^{\prime}}\right)
\left(  \mathcal{R}_{C\rightarrow AC}\left(  \mathcal{T}_{C^{\prime
}\rightarrow C}^{\mathcal{W}}\left(  \mathcal{V}_{B\rightarrow B^{\prime}%
}\otimes\mathcal{W}_{C\rightarrow C^{\prime}}\right)  \left(  \rho
_{BC}\right)  \right)  \right)  \right)  \\
&  \leq\sup_{\mathcal{R}_{C^{\prime}\rightarrow A^{\prime}C^{\prime}}}F\left(
\sigma_{A^{\prime}B^{\prime}C^{\prime}},\mathcal{R}_{C^{\prime}\rightarrow
A^{\prime}C^{\prime}}\left(  \left(  \mathcal{V}_{B\rightarrow B^{\prime}%
}\otimes\mathcal{W}_{C\rightarrow C^{\prime}}\right)  \left(  \rho
_{BC}\right)  \right)  \right)  \\
&  =F\left(  A^{\prime};B^{\prime}|C^{\prime}\right)  _{\sigma}.
\end{align}

The first equality follows from invariance of fidelity with respect to
isometries. The second equality follows because $\mathcal{R}_{C\rightarrow
AC}$ and $\mathcal{V}_{B\rightarrow B^{\prime}}$ commute. The third equality
follows from (\ref{eq:invert-isometry-C}). The inequality follows because%
\begin{equation}
\left(  \mathcal{U}_{A\rightarrow A^{\prime}}\otimes\mathcal{W}_{C\rightarrow
C^{\prime}}\right)  \circ\mathcal{R}_{C\rightarrow AC}\circ\mathcal{T}%
_{C^{\prime}\rightarrow C}^{\mathcal{W}}%
\end{equation}
is a particular CPTP\ recovery map from $C^{\prime}$ to $A^{\prime}C^{\prime}%
$. The last equality is from the definition of fidelity of recovery. Given
that the inequality%
\begin{equation}
F\left(  \rho_{ABC},\mathcal{R}_{C\rightarrow AC}\left(  \rho_{BC}\right)
\right)  \leq F\left(  A^{\prime};B^{\prime}|C^{\prime}\right)  _{\sigma}%
\end{equation}
holds for an arbitrary recovery map $\mathcal{R}_{C\rightarrow AC}$, we can
conclude that $F( A;B|C) _{\rho}\leq F\left(  A^{\prime};B^{\prime}|C^{\prime
}\right)  _{\sigma}. $

For the other inequality, let $\mathcal{R}_{C^{\prime}\rightarrow A^{\prime
}C^{\prime}}$ be an arbitrary recovery map. Then
\begin{align}
&  F\left(  \sigma_{A^{\prime}B^{\prime}C^{\prime}},\mathcal{R}_{C^{\prime
}\rightarrow A^{\prime}C^{\prime}}\left(  \sigma_{B^{\prime}C^{\prime}%
}\right)  \right) \nonumber\\
&  \leq F\left(  \left(  \mathcal{T}_{A^{\prime}\rightarrow A}^{\mathcal{U}%
}\otimes\mathcal{T}_{B^{\prime}\rightarrow B}^{\mathcal{V}}\otimes
\mathcal{T}_{C^{\prime}\rightarrow C}^{\mathcal{W}}\right)  \left(
\sigma_{A^{\prime}B^{\prime}C^{\prime}}\right)  ,\left(  \mathcal{T}%
_{A^{\prime}\rightarrow A}^{\mathcal{U}}\otimes\mathcal{T}_{B^{\prime
}\rightarrow B}^{\mathcal{V}}\otimes\mathcal{T}_{C^{\prime}\rightarrow
C}^{\mathcal{W}}\right)  \left(  \mathcal{R}_{C^{\prime}\rightarrow A^{\prime
}C^{\prime}}\left(  \sigma_{B^{\prime}C^{\prime}}\right)  \right)  \right) \\
&  =F\left(  \rho_{ABC},\left(  \mathcal{T}_{A^{\prime}\rightarrow
A}^{\mathcal{U}}\otimes\mathcal{T}_{C^{\prime}\rightarrow C}^{\mathcal{W}%
}\right)  \left(  \mathcal{R}_{C^{\prime}\rightarrow A^{\prime}C^{\prime}%
}\left(  \mathcal{T}_{B^{\prime}\rightarrow B}^{\mathcal{V}}\left(
\sigma_{B^{\prime}C^{\prime}}\right)  \right)  \right)  \right) \\
&  =F\left(  \rho_{ABC},\left(  \mathcal{T}_{A^{\prime}\rightarrow
A}^{\mathcal{U}}\otimes\mathcal{T}_{C^{\prime}\rightarrow C}^{\mathcal{W}%
}\right)  \left(  \mathcal{R}_{C^{\prime}\rightarrow A^{\prime}C^{\prime}%
}\left(  \left(  \mathcal{T}_{B^{\prime}\rightarrow B}^{\mathcal{V}}%
\circ\mathcal{V}_{B\rightarrow B^{\prime}}\otimes\mathcal{W}_{C\rightarrow
C^{\prime}}\right)  \left(  \rho_{BC}\right)  \right)  \right)  \right) \\
&  =F\left(  \rho_{ABC},\left(  \mathcal{T}_{A^{\prime}\rightarrow
A}^{\mathcal{U}}\otimes\mathcal{T}_{C^{\prime}\rightarrow C}^{\mathcal{W}%
}\right)  \left(  \mathcal{R}_{C^{\prime}\rightarrow A^{\prime}C^{\prime}%
}\left(  \mathcal{W}_{C\rightarrow C^{\prime}}\left(  \rho_{BC}\right)
\right)  \right)  \right) \\
&  \leq\sup_{\mathcal{R}_{C\rightarrow AC}}F\left(  \rho_{ABC},\mathcal{R}%
_{C\rightarrow AC}\left(  \rho_{BC}\right)  \right) \\
&  =F(  A;B|C)  _{\rho}.
\end{align}

The first inequality is from monotonicity of the fidelity with respect to
quantum channels. The first equality is a consequence of
(\ref{eq:invert-isometry-A})-(\ref{eq:invert-isometry-C}). The second equality
is from the definition of $\sigma_{B^{\prime}C^{\prime}}$. The third equality
follows from (\ref{eq:invert-isometry-C}). The last inequality follows because
$\left(  \mathcal{T}_{A^{\prime}\rightarrow A}^{\mathcal{U}}\otimes
\mathcal{T}_{C^{\prime}\rightarrow C}^{\mathcal{W}}\right)  \circ
\mathcal{R}_{C^{\prime}\rightarrow A^{\prime}C^{\prime}}\circ\mathcal{W}%
_{C\rightarrow C^{\prime}}$ is a particular recovery map from $C$ to $AC$.
Given that the inequality
\begin{equation}
F\left(  \sigma_{A^{\prime}B^{\prime}C^{\prime}%
},\mathcal{R}_{C^{\prime}\rightarrow A^{\prime}C^{\prime}}\left(
\sigma_{B^{\prime}C^{\prime}}\right)  \right)  \leq F( A;B|C) _{\rho}
\end{equation} holds
for an arbitrary recovery map $\mathcal{R}_{C^{\prime}\rightarrow A^{\prime
}C^{\prime}}$, we can conclude that $F\left(  A^{\prime};B^{\prime}|C^{\prime
}\right)  _{\sigma}\leq F\left(  A;B|C\right)  _{\rho}. $
\end{proof}

\begin{remark}
The only property of the fidelity used to prove
Propositions~\ref{prop:FoR-mono-local-ops} and \ref{prop:FoR-local-iso} is
that it is monotone with respect to quantum operations. This suggests that we
can construct a fidelity-of-recovery-like measure from any \textquotedblleft
generalized divergence\textquotedblright\ (a function that is monotone with
respect to quantum operations).
\end{remark}

\begin{figure*}[ptb]
\center
\includegraphics[scale=1.0]{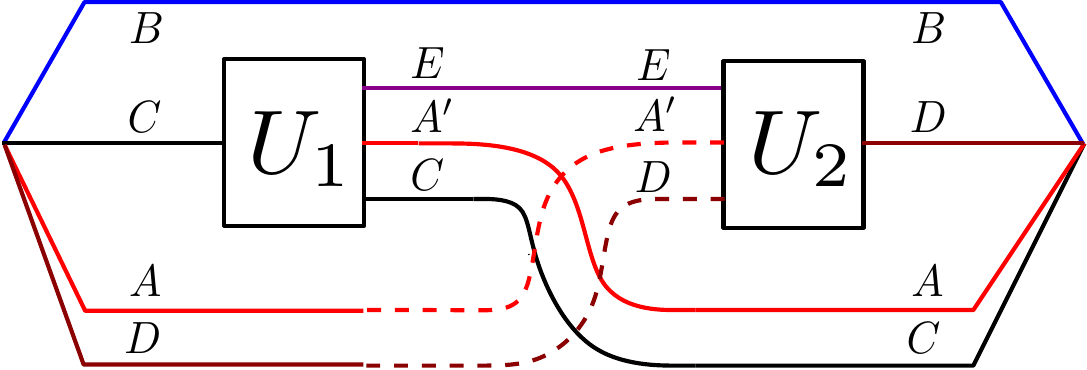}\caption{This figure helps to
illustrate the main idea behind the proof of Proposition~\ref{prop:duality-f}
and furthermore highlights the dual role played by an isometric extension of
the recovery map on $C$ and an Uhlmann isometry acting on system $D$ (and vice
versa). When reading the figure from left to right, the isometry on the left
corresponds to the recovery map and the isometry on the right corresponds to
the one from Uhlmann's theorem, and the overlap between the left and right is
understood as $F(A;B|C)$. When reading the figure from right to left, the
isometries switch their roles and the overlap is understood as $F(A;B|D)$.
This picture clarifies in a diagrammatic way how we get the duality relation
$F(A;B|C) = F(A;B|D)$.}%
\label{fig:duality}%
\end{figure*}

\begin{proposition}
[Duality]\label{prop:duality-f}Let $\phi_{ABCD}$ be a four-party pure state.
Then%
\begin{equation}
F( A;B|C) _{\phi}=F( A;B|D) _{\phi},
\end{equation}
which is equivalent to%
\begin{equation}
I_{F}( A;B|C) _{\phi}=I_{F}( A;B|D) _{\phi}.
\end{equation}

\end{proposition}

\begin{proof}
By definition,%
\begin{equation}
F( A;B|C) _{\phi}=\sup_{\mathcal{R}_{C\rightarrow AC}^{1}}F\left(  \phi
_{ABC},\mathcal{R}_{C\rightarrow AC}^{1}\left(  \phi_{BC}\right)  \right)  .
\end{equation}
Let $\mathcal{U}_{C\rightarrow ACE}^{\mathcal{R}^{1}}$ be an isometric channel
which extends $\mathcal{R}_{C\rightarrow AC}^{1}$. Since $\phi_{ABCD}$ is a
purification of $\phi_{ABC}$ and $\mathcal{U}_{C\rightarrow ACE}%
^{\mathcal{R}^{1}}\left(  \phi_{BCA^{\prime}D}\right)  $ is a purification of
$\mathcal{R}_{C\rightarrow AC}^{1}\left(  \phi_{BC}\right)  $, we can apply
Uhlmann's theorem for fidelity to conclude that
\begin{equation}
\sup_{\mathcal{R}_{C\rightarrow AC}^{1}}F\left(  \phi_{ABC},\mathcal{R}%
_{C\rightarrow AC}^{1}\left(  \phi_{BC}\right)  \right)  =\label{eq:fabc}
\sup_{\mathcal{U}_{D\rightarrow A^{\prime}DE}}\sup_{\mathcal{U}_{C\rightarrow
ACE}^{\mathcal{R}^{1}}}F\left(  \mathcal{U}_{D\rightarrow A^{\prime}DE}\left(
\phi_{ABCD}\right)  ,\mathcal{U}_{C\rightarrow ACE}^{\mathcal{R}^{1}}\left(
\phi_{BCA^{\prime}D}\right)  \right)  .
\end{equation}
Now consider that
\begin{equation}
F(  A;B|D)  _{\phi}=\sup_{\mathcal{R}_{D\rightarrow AD}^{2}%
}F\left(  \phi_{ABD},\mathcal{R}_{D\rightarrow AD}^{2}\left(  \phi
_{BD}\right)  \right)  .
\end{equation}
Let $\mathcal{U}_{D\rightarrow ADE}^{\mathcal{R}^{2}}$ be an isometric channel
which extends $\mathcal{R}_{D\rightarrow AD}^{2}$. Since $\phi_{ABCD}$ is a
purification of $\phi_{ABD}$ and $\mathcal{U}_{D\rightarrow ADE}%
^{\mathcal{R}^{2}}\left(  \phi_{BDA^{\prime}C}\right)  $ is a purification of
$\mathcal{R}_{D\rightarrow AD}^{2}\left(  \phi_{BD}\right)  $, we can apply
Uhlmann's theorem for fidelity to conclude that
\begin{equation}
\sup_{\mathcal{R}_{D\rightarrow AD}^{2}}F\left(  \phi_{ABD},\mathcal{R}%
_{D\rightarrow AD}^{2}\left(  \phi_{BD}\right)  \right)  =\label{eq:fabd}
\sup_{\mathcal{U}_{C\rightarrow A^{\prime}CE}}\sup_{\mathcal{U}_{D\rightarrow
ADE}^{\mathcal{R}^{2}}}F\left(  \mathcal{U}_{C\rightarrow A^{\prime}CE}\left(
\phi_{ABCD}\right)  ,\mathcal{U}_{D\rightarrow ADE}^{\mathcal{R}^{2}}\left(
\phi_{BDA^{\prime}C}\right)  \right)  .
\end{equation}

By inspecting the RHS\ of (\ref{eq:fabc}) and the RHS of (\ref{eq:fabd}), we
see that the two expressions are equivalent so that the statement of the
proposition holds. Figure~\ref{fig:duality} gives a graphical depiction of
this proof which should help in determining which systems are
\textquotedblleft connected together\textquotedblright\ and furthermore
highlights how the duality between the recovery map and the map from Uhlmann's
theorem is reflected in the duality for the fidelity of recovery.
\end{proof}

\begin{remark}
The physical interpretation of the above duality is as follows: beginning with
a four-party pure state $\phi_{ABCD}$, suppose that system $A$ is lost. Then
one can recover the state on systems $ABC\ $from system $C$ alone just as well
as one can recover the state on systems $ABD$ from system $D$ alone.
\end{remark}

\begin{proposition}
[Continuity]\label{prop:FoR-continuity}The fidelity of recovery is a
continuous function of its input. That is, given two tripartite states
$\rho_{ABC}$ and $\sigma_{ABC}$ such that $F\left(  \rho_{ABC},\sigma
_{ABC}\right)  \geq1-\varepsilon$ where $\varepsilon\in\left[  0,1\right]  $,
it follows that%
\begin{align}
\left\vert F( A;B|C) _{\rho}-F( A;B|C) _{\sigma}\right\vert  &  \leq
8\sqrt{\varepsilon},\label{eq:FoR-continuity}\\
\left\vert I_{F}( A;B|C) _{\rho}-I_{F}( A;B|C) _{\sigma}\right\vert  &
\leq\left\vert A\right\vert ^{x}8\sqrt{\varepsilon},
\label{eq:I-FoR-continuity}%
\end{align}
where $x=1$ if system $A$ is classical and $x=2$ otherwise.
\end{proposition}

\begin{proof}
One of the main tools for our proof is the purified distance
\cite[Definition~4]{TCR10}, defined for two quantum states as%
\begin{equation}
P\left(  \rho,\sigma\right)  \equiv\sqrt{1-F\left(  \rho,\sigma\right)  },
\end{equation}
and which for our case implies that%
\begin{equation}
P\left(  \rho_{ABC},\sigma_{ABC}\right)  \leq\sqrt{\varepsilon}.
\end{equation}
From the monotonicity of the purified distance with respect to quantum
operations \cite[Lemma 7]{TCR10}, it follows that%
\begin{equation}
P\left(  \mathcal{R}_{C\rightarrow AC}\left(  \rho_{BC}\right)  ,\mathcal{R}%
_{C\rightarrow AC}\left(  \sigma_{BC}\right)  \right)  \leq\sqrt{\varepsilon},
\end{equation}
where $\mathcal{R}_{C\rightarrow AC}$ is an arbitrary CPTP\ linear recovery
map. By the triangle inequality for purified distance \cite[Lemma 5]{TCR10},
it follows that%
\begin{align}
  \inf_{\mathcal{R}_{C\rightarrow AC}}P\left(  \rho_{ABC},\mathcal{R}%
_{C\rightarrow AC}\left(  \rho_{BC}\right)  \right) 
&  \leq P\left(  \rho_{ABC},\mathcal{R}_{C\rightarrow AC}\left(  \rho
_{BC}\right)  \right) \\
&  \leq P\left(  \rho_{ABC},\sigma_{ABC}\right)  +P\left(  \sigma
_{ABC},\mathcal{R}_{C\rightarrow AC}\left(  \sigma_{BC}\right)  \right)
\nonumber\\
&  \ \ \ \ \ +P\left(  \mathcal{R}_{C\rightarrow AC}\left(  \sigma
_{BC}\right)  ,\mathcal{R}_{C\rightarrow AC}\left(  \rho_{BC}\right)  \right)
\\
&  \leq2\sqrt{\varepsilon}+P\left(  \sigma_{ABC},\mathcal{R}_{C\rightarrow
AC}\left(  \sigma_{BC}\right)  \right)  .
\end{align}
Given that $\mathcal{R}_{C\rightarrow AC}$ is arbitrary, we can conclude that%
\begin{equation}
\inf_{\mathcal{R}_{C\rightarrow AC}}P\left(  \rho_{ABC},\mathcal{R}%
_{C\rightarrow AC}\left(  \rho_{BC}\right)  \right) 
\leq2\sqrt{\varepsilon}+\inf_{\mathcal{R}_{C\rightarrow AC}}P\left(
\sigma_{ABC},\mathcal{R}_{C\rightarrow AC}\left(  \sigma_{BC}\right)  \right)
,
\end{equation}
which is equivalent to%
\begin{equation}
\sqrt{1-F( A;B|C) _{\rho}}\leq2\sqrt{\varepsilon}+\sqrt{1-F( A;B|C) _{\sigma}%
}.
\end{equation}
Squaring both sides gives%
\begin{align}
  1-F( A;B|C) _{\rho}
&  \leq4\varepsilon+4\sqrt{\varepsilon}\sqrt{1-F( A;B|C) _{\sigma}}+1-F(
A;B|C) _{\sigma}\nonumber\\
&  \leq8\sqrt{\varepsilon}+1-F( A;B|C) _{\sigma},
\end{align}
where the second inequality follows because $\varepsilon\in\left[  0,1\right]
$ and the same is true for the fidelity. Rewriting this gives%
\begin{equation}
F( A;B|C) _{\sigma}\leq8\sqrt{\varepsilon}+F( A;B|C) _{\rho}.
\label{eq:FoR-continuity-1}%
\end{equation}
The same approach gives the other inequality:%
\begin{equation}
F( A;B|C) _{\rho}\leq8\sqrt{\varepsilon}+F( A;B|C) _{\sigma}.
\label{eq:FoR-continuity-2}%
\end{equation}

By dividing (\ref{eq:FoR-continuity-1}) by $F( A;B|C) _{\rho}$ (which by
Proposition~\ref{prop:dim-bound}\ is never smaller than $1/\left\vert
A\right\vert ^{2}$) and taking a logarithm, we find that%
\begin{align}
\log\left(  \frac{F( A;B|C) _{\sigma}}{F( A;B|C) _{\rho}}\right)   &  \leq
\log\left(  1+\frac{8\sqrt{\varepsilon}}{F\left(  A;B|C\right)  _{\rho}%
}\right) \\
&  \leq\frac{8\sqrt{\varepsilon}}{F( A;B|C) _{\rho}}\\
&  \leq\left\vert A\right\vert ^{x}8\sqrt{\varepsilon}.
\end{align}
where we used that $\log\left(  y+1\right)  \leq y$ and the dimension bound
from Proposition~\ref{prop:dim-bound}. Applying this to the other inequality
in (\ref{eq:FoR-continuity-2}) gives that%
\begin{equation}
\log\left(  \frac{F( A;B|C) _{\rho}}{F( A;B|C) _{\sigma}}\right)
\leq\left\vert A\right\vert ^{x}8\sqrt{\varepsilon},
\end{equation}
from which we can conclude (\ref{eq:I-FoR-continuity}).
\end{proof}

\begin{proposition}
[Weak chain rule]Given a four-party state $\rho_{ABCD}$, the following
inequality holds%
\begin{equation}
I_{F}\left(  AC;B|D\right)  _{\rho}\geq I_{F}\left(  A;B|CD\right)  _{\rho}.
\end{equation}

\end{proposition}

\begin{proof}
The inequality is equivalent to%
\begin{equation}
F\left(  AC;B|D\right)  _{\rho}\leq F\left(  A;B|CD\right)  _{\rho},
\label{eq:weak-chain}%
\end{equation}
which follows from the fact that it is easier to recover $A$ from $CD$ than it
is to recover both $A$ and $C$ from $D$ alone. Indeed, let $\mathcal{R}%
_{D\rightarrow ACD}$ be any recovery map. Then%
\begin{align}
  F\left(  \rho_{ABCD},\mathcal{R}_{D\rightarrow ACD}\left(  \rho
_{BD}\right)  \right) 
&  =F\left(  \rho_{ABCD},\left(  \mathcal{R}_{D\rightarrow ACD}\circ
\operatorname{Tr}_{C}\right)  \left(  \rho_{BCD}\right)  \right) \\
&  \leq\sup_{\mathcal{R}_{CD\rightarrow ACD}}F\left(  \rho_{ABCD},\left(
\mathcal{R}_{CD\rightarrow ACD}\right)  \left(  \rho_{BCD}\right)  \right) \\
&  =F\left(  A;B|CD\right)  _{\rho}.
\end{align}
Since the chain of inequalities holds for any recovery map $\mathcal{R}%
_{D\rightarrow ACD}$, we can conclude (\ref{eq:weak-chain}) from the
definition of $F\left(  AC;B|D\right)  _{\rho}$.
\end{proof}

\begin{proposition}
[Conditioning on classical info.]\label{prop:condition-classical}Let
$\omega_{ABCX}$ be a state for which system $X$ is classical:%
\begin{equation}
\omega_{ABCX}=\sum_{x}p_{X}( x) \omega_{ABC}^{x}\otimes\left\vert
x\right\rangle \left\langle x\right\vert _{X},
\end{equation}
where $\left\{  \left\vert x\right\rangle _{X}\right\}  $ is an orthonormal
basis, $p_{X}$ is a probability distribution, and each $\omega_{ABC}^{x}$ is a
state. Then the following inequalities hold%
\begin{align}
\sqrt{F}\left(  A;B|CX\right)  _{\omega}  &  \geq\sum_{x}p_{X}\left(
x\right)  \sqrt{F}( A;B|C) _{\omega^{x}},\label{eq:FoR-conditioning-classical}%
\\
I_{F}\left(  A;B|CX\right)  _{\omega}  &  \leq\sum_{x}p_{X}( x) I_{F}( A;B|C)
_{\omega^{x}}. \label{eq:IF-conditioning-classical}%
\end{align}

\end{proposition}

\begin{proof}
We first prove the inequality in (\ref{eq:FoR-conditioning-classical}). For
any set of recovery maps $\mathcal{R}_{C\rightarrow CA}^{x}$, define
$\mathcal{R}_{CX\rightarrow CXA}$ as follows:%
\begin{equation}
\mathcal{R}_{CX\rightarrow CXA}\left(  \tau_{CX}\right)  \equiv
\sum_{x}\mathcal{R}_{C\rightarrow CA}^{x}\left(  \left\langle x\right\vert
_{X}\left(  \tau_{CX}\right)  \left\vert x\right\rangle _{X}\right)
\ \left\vert x\right\rangle \left\langle x\right\vert _{X},
\end{equation}
so that it first measures the system $X$ in the basis $\left\{  \left\vert
x\right\rangle \left\langle x\right\vert _{X}\right\}  $, places the outcome
in the same classical register, and then acts with the particular recovery map
$\mathcal{R}_{C\rightarrow CA}^{x}$. Then 
\begin{align}
&  \left[  \sum_{x}p_{X}(  x)  \sqrt{F}\left(  \omega_{ABC}%
^{x},\mathcal{R}_{C\rightarrow CA}^{x}\left(  \omega_{BC}^{x}\right)  \right)
\right]  ^{2}\nonumber\\
&  =F\left(  \sum_{x}p_{X}(  x)  \omega_{ABC}^{x}\otimes\left\vert
x\right\rangle \left\langle x\right\vert _{X},\sum_{x}p_{X}(  x)
\mathcal{R}_{C\rightarrow CA}^{x}\left(  \omega_{BC}^{x}\right)
\otimes\left\vert x\right\rangle \left\langle x\right\vert _{X}\right)  \\
&  =F\left(  \sum_{x}p_{X}(  x)  \omega_{ABC}^{x}\otimes\left\vert
x\right\rangle \left\langle x\right\vert _{X},\mathcal{R}_{CX\rightarrow
CXA}\left(  \sum_{x}p_{X}(  x)  \omega_{BC}^{x}\otimes\left\vert
x\right\rangle \left\langle x\right\vert _{X}\right)  \right)  \\
&  \leq F\left(  A;B|CX\right)  _{\omega}.
\end{align}

Since the inequality holds for any set of individual recovery maps $\left\{
\mathcal{R}_{C\rightarrow CA}^{x}\right\}  $, we obtain
(\ref{eq:FoR-conditioning-classical}).

Finally, we recover (\ref{eq:IF-conditioning-classical}) by applying a
negative logarithm to the inequality in (\ref{eq:FoR-conditioning-classical})
and exploiting convexity of $-\log$.
\end{proof}

\begin{proposition}
[Conditioning on a product system]\label{prop:condition-on-product}Let
$\rho_{ABC}=\sigma_{AB}\otimes\omega_{C}$. Then%
\begin{align}
F( A;B|C) _{\rho}  &  =F( A;B) _{\sigma}\equiv\sup_{\tau_{A}}F\left(
\sigma_{AB},\tau_{A}\otimes\sigma_{B}\right)  ,\\
I_{F}( A;B|C) _{\rho}  &  =I_{F}( A;B) _{\sigma}\equiv-\log F( A;B) _{\sigma}.
\label{eq:sandwich-similar}%
\end{align}

\end{proposition}

\begin{proof}
Consider that, for any recovery map $\mathcal{R}_{C\rightarrow AC}$%
\begin{align}
  F\left(  \sigma_{AB}\otimes\omega_{C},\mathcal{R}_{C\rightarrow AC}\left(
\sigma_{B}\otimes\omega_{C}\right)  \right) 
&  =F\left(  \sigma_{AB}\otimes\omega_{C},\sigma_{B}\otimes\mathcal{R}%
_{C\rightarrow AC}\left(  \omega_{C}\right)  \right) \\
&  \leq F\left(  \sigma_{AB},\sigma_{B}\otimes\mathcal{R}_{C\rightarrow
A}\left(  \omega_{C}\right)  \right) \\
&  \leq\sup_{\tau_{A}}F\left(  \sigma_{AB},\sigma_{B}\otimes\tau_{A}\right)  .
\end{align}
The first inequality follows because fidelity is monotone with respect to a
partial trace over the $C$ system. The second inequality follows by optimizing
the second argument to the fidelity over all states on the $A$ system. Since
the inequality holds independent of the recovery map $\mathcal{R}%
_{C\rightarrow AC}$, we find that%
\begin{equation}
F( A;B|C) _{\rho}\leq F( A;B) _{\sigma}.
\end{equation}
To prove the other inequality $F( A;B) _{\sigma}\leq F\left(  A;B|C\right)
_{\rho}$, consider for any state $\tau_{A}$ that%
\begin{align}
  F\left(  \sigma_{AB},\tau_{A}\otimes\sigma_{B}\right)
&  =F\left(  \sigma_{AB}\otimes\omega_{C},\tau_{A}\otimes\sigma_{B}%
\otimes\omega_{C}\right) \\
&  =F\left(  \sigma_{AB}\otimes\omega_{C},\left(  \mathcal{P}_{A}^{\tau
}\otimes\text{id}_{C}\right)  \left(  \sigma_{B}\otimes\omega_{C}\right)
\right) \\
&  \leq\sup_{\mathcal{R}_{C\rightarrow AC}}F\left(  \sigma_{AB}\otimes
\omega_{C},\mathcal{R}_{C\rightarrow AC}\left(  \sigma_{B}\otimes\omega
_{C}\right)  \right)  .
\end{align}
The first equality follows because fidelity is multiplicative with respect to
tensor-product states. The second equality follows by taking $\left(
\text{id}_{C}\otimes\mathcal{P}_{A}^{\tau}\right)  $ to be the recovery map
that does nothing to system $C$ and prepares $\tau_{A}$ on system $A$. The
inequality follows by optimizing over all recovery maps. Since the inequality
is independent of the prepared state, we obtain the other inequality%
\begin{equation}
F( A;B) _{\sigma}\leq F( A;B|C) _{\rho}.
\end{equation}

The equality $I_{F}(A;B|C)_{\rho}=I_{F}(A;B)_{\sigma}$ follows by applying a
negative logarithm to $F(A;B|C)_{\rho}=F(A;B)_{\sigma}$. We note in passing
that the quantity on the RHS\ in (\ref{eq:sandwich-similar}) is closely
related to the sandwiched R\'{e}nyi mutual information of order 1/2
\cite{MDSFT13,WWY13,B13,GW13}.
\end{proof}

\begin{proposition}
[Dimension bound]\label{prop:dim-bound}The fidelity of recovery obeys the
following dimension bound:%
\begin{equation}
F( A;B|C) _{\rho}\geq\frac{1}{\left\vert A\right\vert ^{2}},
\label{eq:FoR-dim-bnd-1}%
\end{equation}
which is equivalent to%
\begin{equation}
I_{F}( A;B|C) _{\rho}\leq2\log\left\vert A\right\vert .
\label{eq:FoR-dim-bnd-2}%
\end{equation}
If the system $A$ is classical, so that we relabel it as $X$, then the
following hold%
\begin{align}
F( X;B|C) _{\rho}  &  \geq\frac{1}{\left\vert X\right\vert }%
,\label{eq:FoR-class-dim-bnd-1}\\
I_{F}( X;B|C) _{\rho}  &  \leq\log\left\vert X\right\vert .
\label{eq:FoR-class-dim-bnd-2}%
\end{align}
Examples of states achieving these bounds are $\Phi_{AB}\otimes\sigma_{C}$ for
\eqref{eq:FoR-dim-bnd-1}-\eqref{eq:FoR-dim-bnd-2} and $\overline{\Phi}%
_{XB}\otimes\sigma_{C}$ for
\eqref{eq:FoR-class-dim-bnd-1}-\eqref{eq:FoR-class-dim-bnd-2}, where%
\begin{equation}
\overline{\Phi}_{XB}\equiv\frac{1}{\left\vert X\right\vert }\sum_{x}\left\vert
x\right\rangle \left\langle x\right\vert _{X}\otimes\left\vert x\right\rangle
\left\langle x\right\vert _{B}.
\end{equation}

\end{proposition}

\begin{proof}
Consider that the following inequality holds, simply by choosing the recovery
map to be one in which we do not do anything to system $C$ and prepare the
maximally mixed state $\pi_{A}\equiv I_{A}/\left\vert A\right\vert $\ on
system $A$:%
\begin{align}
F( A;B|C) _{\rho}  &  \geq F\left(  \rho_{ABC},\pi_{A}\otimes\rho_{BC}\right)
\\
&  =\frac{1}{\left\vert A\right\vert }F\left(  \rho_{ABC},I_{A}\otimes
\rho_{BC}\right) \\
&  \geq\frac{1}{\left\vert A\right\vert }\left[  \operatorname{Tr}\left\{
\sqrt{\rho_{ABC}}\sqrt{I_{A}\otimes\rho_{BC}}\right\}  \right]  ^{2}.
\end{align}
Taking a negative logarithm and letting $\phi_{ABCD}$ be a purification of
$\rho_{ABC}$, we find that%
\begin{align}
  I_{F}( A;B|C) _{\rho}
&  \leq\log\left\vert A\right\vert -2\log\operatorname{Tr}\left\{  \sqrt
{\rho_{ABC}}\sqrt{I_{A}\otimes\rho_{BC}}\right\} \\
&  =\log\left\vert A\right\vert -H_{1/2}( A|BC) _{\rho}\label{eq:branch-step}%
\\
&  =\log\left\vert A\right\vert +H_{3/2}( A|D) _{\rho}\\
&  \leq\log\left\vert A\right\vert +H_{3/2}( A) _{\rho}\\
&  \leq2\log\left\vert A\right\vert .
\end{align}
The first equality follows by recognizing that the second term is a
conditional R\'{e}nyi entropy of order 1/2 \cite[Definition~3]{TCR09} (see Appendix~\ref{app} for a definition). The
second equality follows from a duality relation for this conditional R\'{e}nyi
entropy \cite[Lemma 6]{TCR09}. The second inequality is a consequence of the
quantum data processing inequality for conditional R\'{e}nyi entropies
\cite[Lemma 5]{TCR09} (with the map taken to be a partial trace over system
$D$). The last inequality follows from a dimension bound which holds for any
R\'{e}nyi entropy.

To see that $\Phi_{AB}\otimes\sigma_{C}$ has $I_{F}( A;B|C) =2\log\left\vert
A\right\vert $, we can apply Propositions~\ref{prop:normalization} and
\ref{prop:pure-state}.

For classical $A$ system, we follow the same steps up to (\ref{eq:branch-step}%
), but then apply Lemma~\ref{lem:classical-non-neg} in Appendix~\ref{app} to conclude that
$H_{1/2}(A|BC)\geq0$ for a classical $A$. This gives
\eqref{eq:FoR-class-dim-bnd-1}-\eqref{eq:FoR-class-dim-bnd-2}. To see that
$\overline{\Phi}_{XB}\otimes\sigma_{C}$ has $I_{F}(X;B|C)=\log\left\vert
X\right\vert $, we apply Proposition~\ref{prop:condition-on-product} and then
evaluate%
\begin{align}
  F\left(  \overline{\Phi}_{XB},\tau_{X}\otimes\overline{\Phi}_{B}\right)
&  =\left\Vert \left(  \sum_{x}\frac{1}{\sqrt{\left\vert X\right\vert }%
}\left\vert x\right\rangle \left\langle x\right\vert _{X}\otimes\left\vert
x\right\rangle \left\langle x\right\vert _{B}\right)  \left(  \sqrt{\tau_{X}%
}\otimes\frac{1}{\sqrt{\left\vert X\right\vert }}I_{B}\right)  \right\Vert
_{1}^{2}\nonumber\\
&  =\left[  \frac{1}{\left\vert X\right\vert }\left\Vert \left(  \sum
_{x}\left\vert x\right\rangle \left\langle x\right\vert _{X}\otimes\left\vert
x\right\rangle \left\langle x\right\vert _{B}\right)  \left(  \sqrt{\tau_{X}%
}\otimes I_{B}\right)  \right\Vert _{1}\right]  ^{2}\nonumber\\
&  =\left[  \frac{1}{\left\vert X\right\vert }\sum_{x}\left\Vert \left\vert
x\right\rangle \left\langle x\right\vert _{X}\sqrt{\tau_{X}}\right\Vert
_{1}\right]  ^{2}\nonumber\\
&  =\left[  \frac{1}{\left\vert X\right\vert }\sum_{x}\sqrt{\left\langle
x\right\vert \tau\left\vert x\right\rangle }\right]  ^{2}\nonumber\\
&  \leq\frac{1}{\left\vert X\right\vert }\sum_{x}\left\langle x\right\vert
\tau\left\vert x\right\rangle \nonumber\\
&  =\frac{1}{\left\vert X\right\vert }.
\end{align}
Choosing $\tau_{X}$ maximally mixed then achieves the upper bound, i.e.,
\begin{align}
\sup_{\tau_{X}}F\left(  \overline{\Phi}_{XB},\tau_{X}\otimes\overline{\Phi
}_{B}\right)   &  =F\left(  \overline{\Phi}_{XB},\pi_{X}\otimes\overline{\Phi
}_{B}\right)  \\
&  =\frac{1}{\left\vert X\right\vert }.
\end{align}

\end{proof}

The following proposition gives a simple proof of the main result of
\cite{FR14}\ when the tripartite state of interest is pure:

\begin{proposition}
[Approximate q.~Markov chain]\label{prop:approx-QMC}The conditional
mutual information $I( A;B|C) _{\psi}$ of a pure tripartite state $\psi_{ABC}$
has the following lower bound:%
\begin{equation}
I( A;B|C) _{\psi}\geq-\log F( A;B|C) _{\psi}.
\end{equation}

\end{proposition}

\begin{proof}
Let $\varphi_{D}$ be a pure state on an auxiliary system $D$, so that
$\left\vert \psi\right\rangle _{ABC}\otimes\left\vert \varphi\right\rangle
_{D}$ is a purification of $\left\vert \psi\right\rangle _{ABC}$. Consider the
following chain of inequalities:%
\begin{align}
I( A;B|C) _{\psi}  &  =I( A;B|D) _{\psi\otimes\varphi}\\
&  =I( A;B) _{\psi}\\
&  \geq-\log F\left(  \psi_{AB},\psi_{A}\otimes\psi_{B}\right) \label{eq:monotone-renyi-step}\\
&  \geq-\log F( A;B) _{\psi}\\
&  =-\log F( A;B|D) _{\psi\otimes\varphi}\\
&  =-\log F( A;B|C) _{\psi}.
\end{align}
The first equality follows from duality of conditional mutual information. The
second follows because system $D$ is product with systems $A$ and $B$. The
first inequality follows from monotonicity of the sandwiched R\'{e}nyi
relative entropies \cite[Theorem 7]{MDSFT13}:%
\begin{equation}
\widetilde{D}_{\alpha}(  \rho\Vert\sigma)  \leq\widetilde{D}%
_{\beta}(  \rho\Vert\sigma) \label{eq:renyi-monotone-sandwiched} ,
\end{equation}
for states $\rho$ and $\sigma$ and R\'{e}nyi parameters $\alpha$ and $\beta$
such that $0\leq\alpha\leq\beta$. Recall that the sandwiched R\'enyi relative entropy is defined as \cite{MDSFT13,WWY13}
\begin{equation}
\widetilde{D}_{\alpha}(  \rho\Vert\sigma) \equiv
\frac{2 \alpha}{\alpha-1}
\log \left\Vert \sigma^{(1-\alpha)/2\alpha} \rho^{1/2}
\right\Vert_{2 \alpha}
\end{equation}
whenever $\operatorname{supp}(\rho) \subseteq 
\operatorname{supp}(\sigma)$, and it is equal to
$+\infty$ otherwise. The following limit is known
\cite{MDSFT13,WWY13}:
\begin{equation}
\lim_{\alpha \to 1}\widetilde{D}_{\alpha}(  \rho\Vert\sigma) = D(  \rho\Vert\sigma),
\end{equation}
where the quantum relative entropy is defined as
$D(\rho\Vert\sigma) \equiv \operatorname{Tr}
\{\rho [\log \rho - \log \sigma]  \}$ whenever 
$\operatorname{supp}(\rho) \subseteq 
\operatorname{supp}(\sigma)$, and it is equal to
$+\infty$ otherwise.
To arrive at \eqref{eq:monotone-renyi-step}, we apply
\eqref{eq:renyi-monotone-sandwiched} with the choices $\alpha=1/2$,
$\beta=1$, $\rho=\psi_{AB}$, and $\sigma=\psi_{A}\otimes\psi_{B}$. The second
inequality follows by optimizing over states on system $A$ and applying the
definition in (\ref{eq:sandwich-similar}). The second-to-last equality follows
from Proposition~\ref{prop:condition-on-product}\ and the last from
Proposition~\ref{prop:duality-f}.
\end{proof}

\section{Geometric squashed entanglement}

In this section, we formally define the geometric squashed entanglement of a
bipartite state $\rho_{AB}$, and we prove that it obeys the properties claimed
in Section~\ref{sec:summary}.

\begin{definition}
[Geometric squashed entanglement]\label{def:gse}The geometric squashed
entanglement of a bipartite state $\rho_{AB}$ is defined as follows:%
\begin{equation}
E_{F}^{\operatorname{sq}}( A;B) _{\rho}\equiv-\frac{1}{2}\log
F^{\operatorname{sq}}( A;B) _{\rho},
\end{equation}
where%
\begin{align}
  F^{\operatorname{sq}}( A;B) _{\rho}
&  \equiv\sup_{\omega_{ABE}}\left\{  F( A;B|E) _{\rho}:\rho_{AB}%
=\operatorname{Tr}_{E}\left\{  \omega_{ABE}\right\}  \right\}
\end{align}

\end{definition}

The geometric squashed entanglement can equivalently be written in terms of an
optimization over \textquotedblleft squashing channels\textquotedblright%
\ acting on a purifying system of the original state (cf.~\cite[Eq.~(3)]{CW04}):

\begin{proposition}
Let $\rho_{AB}$ be a bipartite state and let $\left\vert \psi\right\rangle
_{ABE^{\prime}}$ be a fixed purification of it. Then%
\begin{equation}
F^{\operatorname{sq}}( A;B) _{\rho}=\sup_{\mathcal{S}_{E^{\prime}\rightarrow
E}}F( A;B|E) _{\mathcal{S}\left(  \psi\right)  },
\end{equation}
where the optimization is over quantum channels $\mathcal{S}_{E^{\prime
}\rightarrow E}$.
\end{proposition}

\begin{proof}
We first prove the inequality $F^{\operatorname{sq}}( A;B) _{\rho}\geq
\sup_{\mathcal{S}_{E^{\prime}\rightarrow E}}F( A;B|E) _{\mathcal{S}\left(
\psi\right)  }$. Indeed, for a given purification $\psi_{ABE^{\prime}}$ and
squashing channel $\mathcal{S}_{E^{\prime}\rightarrow E}$, the state
$\mathcal{S}_{E^{\prime}\rightarrow E}\left(  \psi_{ABE^{\prime}}\right)  $ is
an extension of $\rho_{AB}$. So it follows by definition that%
\begin{equation}
F( A;B|E) _{\mathcal{S}\left(  \psi\right)  }\leq F^{\operatorname{sq}}( A;B)
_{\rho}.
\end{equation}
Since the choice of squashing channel was arbitrary, the first inequality follows.

We now prove the other inequality%
\begin{equation}
F^{\operatorname{sq}}( A;B) _{\rho}\leq\sup_{\mathcal{S}_{E^{\prime
}\rightarrow E}}F( A;B|E) _{\mathcal{S}\left(  \psi\right)  }.
\label{eq:other-ineq-sq-chan}%
\end{equation}
Let $\omega_{ABE}$ be an extension of $\rho_{AB}$. Let $\varphi_{ABEE_{1}}$ be
a purification of $\omega_{ABE}$, which is in turn also a purification of
$\rho_{AB}$. Since all purifications are related by isometries acting on the
purifying system, we know that there exists an isometry $U_{E^{\prime
}\rightarrow EE_{1}}^{\omega}$ (depending on $\omega$) such that%
\begin{equation}
\left\vert \varphi\right\rangle _{ABEE_{1}}=U_{E^{\prime}\rightarrow EE_{1}%
}^{\omega}\left\vert \psi\right\rangle _{ABE^{\prime}}.
\end{equation}
Furthermore, we know that%
\begin{align}
\omega_{ABE}  &  =\operatorname{Tr}_{E_{1}}\left\{  U_{E^{\prime}\rightarrow
EE_{1}}^{\omega}\psi_{ABE^{\prime}}\left(  U_{E^{\prime}\rightarrow EE_{1}%
}^{\omega}\right)  ^{\dag}\right\} \\
&  \equiv\mathcal{S}_{E^{\prime}\rightarrow E}^{\omega}\left(  \psi
_{ABE^{\prime}}\right)  ,
\end{align}
where we define the squashing channel $\mathcal{S}_{E^{\prime}\rightarrow
E}^{\omega}$ from the isometry $U_{E^{\prime}\rightarrow EE_{1}}^{\omega}$. So
this implies that%
\begin{align}
F( A;B|E) _{\omega}  &  =F( A;B|E) _{\mathcal{S}^{\omega}\left(  \psi\right)
}\\
&  \leq\sup_{\mathcal{S}_{E^{\prime}\rightarrow E}}F( A;B|E) _{\mathcal{S}%
\left(  \psi\right)  }.
\end{align}
Since the inequality above holds for all extensions, the inequality in
(\ref{eq:other-ineq-sq-chan}) follows.
\end{proof}

The following statement is a direct consequence of
Proposition~\ref{prop:FoR-mono-local-ops}:

\begin{corollary}
\label{prop:gse-mono-lo}The geometric squashed entanglement is monotone with
respect to local operations on both systems $A$ and $B$:%
\begin{equation}
E_{F}^{\operatorname{sq}}( A;B) _{\rho}\geq E_{F}^{\operatorname{sq}}\left(
A^{\prime};B^{\prime}\right)  _{\tau}, \label{eq:mono-LO}%
\end{equation}
where $\tau_{A^{\prime}B^{\prime}}\equiv\left(  \mathcal{N}_{A\rightarrow
A^{\prime}}\otimes\mathcal{M}_{B\rightarrow B^{\prime}}\right)  \left(
\rho_{AB}\right)  $ and $\mathcal{N}_{A\rightarrow A^{\prime}}$ and
$\mathcal{M}_{B\rightarrow B^{\prime}}$ are local quantum channels. This is
equivalent to%
\begin{equation}
F^{\operatorname{sq}}( A;B) _{\rho}\leq F^{\operatorname{sq}}\left(
A^{\prime};B^{\prime}\right)  _{\tau}. \label{eq:mono-LO-FoR}%
\end{equation}

\end{corollary}

\begin{proof}
Let $\omega_{ABE}$ be an arbitrary extension of $\rho_{AB}$ and let%
\begin{equation}
\theta_{A^{\prime}B^{\prime}E}\equiv\left(  \mathcal{N}_{A\rightarrow
A^{\prime}}\otimes\mathcal{M}_{B\rightarrow B^{\prime}}\right)  \left(
\omega_{ABE}\right)  .
\end{equation}
Then by the monotonicity of fidelity of recovery with respect to local quantum
operations, we find that%
\begin{equation}
F( A;B|E) _{\omega}\leq F\left(  A^{\prime};B^{\prime}|E\right)  _{\theta}\leq
F^{\operatorname{sq}}\left(  A^{\prime};B^{\prime}\right)  _{\tau}.
\end{equation}
Since the inequality holds for an arbitrary extension $\omega_{ABE}$ of
$\rho_{AB}$, we can conclude that (\ref{eq:mono-LO-FoR}) holds and
(\ref{eq:mono-LO}) follows by definition.
\end{proof}

\begin{proposition}
The geometric squashed entanglement is invariant with respect to local
isometries, in the sense that%
\begin{equation}
E_{F}^{\operatorname{sq}}( A;B) _{\rho}=E_{F}^{\operatorname{sq}}(A^{\prime
};B^{\prime})_{\sigma},
\end{equation}
where%
\begin{equation}
\sigma_{A^{\prime}B^{\prime}}\equiv\left(  \mathcal{U}_{A\rightarrow
A^{\prime}}\otimes\mathcal{V}_{B\rightarrow B^{\prime}}\right)  \left(
\rho_{AB}\right)
\end{equation}
and $\mathcal{U}_{A\rightarrow A^{\prime}}$ and $\mathcal{V}_{B\rightarrow
B^{\prime}}$ are isometric quantum channels.
\end{proposition}

\begin{proof}
From Corollary~\ref{prop:gse-mono-lo}, we can conclude that%
\begin{equation}
E_{F}^{\operatorname{sq}}( A;B) _{\rho}\geq E_{F}^{\operatorname{sq}%
}(A^{\prime};B^{\prime})_{\sigma}.
\end{equation}
Now let $\mathcal{T}_{A^{\prime}\rightarrow A}^{\mathcal{U}}$ and
$\mathcal{T}_{B^{\prime}\rightarrow B}^{\mathcal{V}}$ be the quantum channels
defined in (\ref{eq:T-maps}). Again using Corollary~\ref{prop:gse-mono-lo}, we
find that%
\begin{align}
E_{F}^{\operatorname{sq}}(A^{\prime};B^{\prime})_{\sigma}  &  \geq
E_{F}^{\operatorname{sq}}(A;B)_{\left(  \mathcal{T}^{\mathcal{U}}%
\otimes\mathcal{T}^{\mathcal{V}}\right)  \left(  \sigma\right)  }\\
&  =E_{F}^{\operatorname{sq}}( A;B) _{\rho},
\end{align}
where the equality follows from (\ref{eq:invert-isometry-A}%
)-(\ref{eq:invert-isometry-B}).
\end{proof}

\begin{proposition}
\label{prop:inv-cc}The geometric squashed entanglement obeys the following
classical communication relations:%
\begin{align}
E_{F}^{\operatorname{sq}}\left(  AX_{A};B\right)  _{\rho}  &  \leq
E_{F}^{\operatorname{sq}}\left(  AX_{A};BX_{B}\right)  _{\rho}\\
&  =E_{F}^{\operatorname{sq}}\left(  A;BX_{B}\right)  _{\rho},
\end{align}
for a state $\rho_{X_{A}X_{B}AB}$ defined as%
\begin{equation}
\rho_{X_{A}X_{B}AB}\equiv\sum_{x}p_{X}( x) \left\vert x\right\rangle
\left\langle x\right\vert _{X_{A}}\otimes\left\vert x\right\rangle
\left\langle x\right\vert _{X_{B}}\otimes\rho_{AB}^{x}.
\end{equation}
These are equivalent to%
\begin{align}
F^{\operatorname{sq}}\left(  AX_{A};B\right)  _{\rho}  &  \geq
F^{\operatorname{sq}}\left(  AX_{A};BX_{B}\right)  _{\rho}\\
&  =F^{\operatorname{sq}}\left(  A;BX_{B}\right)  _{\rho}.
\end{align}

\end{proposition}

\begin{proof}
From monotonicity with respect to local operations, we find that%
\begin{align}
F^{\operatorname{sq}}\left(  AX_{A};BX_{B}\right)  _{\rho}  &  \leq
F^{\operatorname{sq}}\left(  AX_{A};B\right)  _{\rho},\\
F^{\operatorname{sq}}\left(  AX_{A};BX_{B}\right)  _{\rho}  &  \leq
F^{\operatorname{sq}}\left(  A;BX_{B}\right)  _{\rho}.
\end{align}
We now give a proof of the following inequality:%
\begin{equation}
F^{\operatorname{sq}}\left(  A;BX_{B}\right)  _{\rho}\leq F^{\operatorname{sq}%
}\left(  AX_{A};BX_{B}\right)  _{\rho}.
\end{equation}
Let%
\begin{equation}
\rho_{X_{A}X_{B}X_{E}ABE}=
\sum_{x}p_{X}( x) \left\vert x\right\rangle \left\langle x\right\vert _{X_{A}%
}\otimes\left\vert x\right\rangle \left\langle x\right\vert _{X_{B}}%
\otimes\left\vert x\right\rangle \left\langle x\right\vert _{X_{E}}\otimes
\rho_{ABE}^{x},
\end{equation}
where $\rho_{ABE}^{x}$ extends $\rho_{AB}^{x}$. Observe that $\rho_{X_{A}%
X_{B}X_{E}ABE}$ is an extension of $\rho_{X_{A}X_{B}AB}$ and $\rho_{X_{B}ABE}$
is an arbitrary extension of $\rho_{X_{B}AB}$. Let $\mathcal{R}_{E\rightarrow
AE}$ be an arbitrary recovery channel and let $\mathcal{R}_{EX_{E}\rightarrow
AX_{A}EX_{E}}$ be a channel that copies the value in $X_{E}$ to $X_{A}$ and
applies $\mathcal{R}_{E\rightarrow AE}$ to system $E$. Consider that

\begin{align}
&  F\left(  \rho_{ABX_{B}E},\mathcal{R}_{E\rightarrow AE}\left(  \rho
_{BX_{B}E}\right)  \right)  \\
&  =\left[  \sum_{x}p_{X}(  x)  \sqrt{F}\left(  \rho_{ABE}%
^{x},\mathcal{R}_{E\rightarrow AE}\left(  \rho_{BE}^{x}\right)  \right)
\right]  ^{2}\\
&  =F\left(  \sum_{x}p_{X}(  x)  \left\vert xxx\right\rangle
\left\langle xxx\right\vert _{X_{A}X_{B}X_{E}}\otimes\rho_{ABE}^{x},\sum
_{x}p_{X}(  x)  \left\vert xxx\right\rangle \left\langle
xxx\right\vert _{X_{A}X_{B}X_{E}}\otimes\mathcal{R}_{E\rightarrow AE}\left(
\rho_{BE}^{x}\right)  \right)  \\
&  =F\left(  \rho_{AX_{A}BX_{B}EX_{E}},\mathcal{R}_{EX_{E}\rightarrow
AX_{A}EX_{E}}\left(  \rho_{BX_{B}EX_{E}}\right)  \right)  \\
&  \leq F^{\text{sq}}\left(  AX_{A};BX_{B}\right)  _{\rho}.
\end{align}

The first two equalities are a consequence of the following property of
fidelity:%
\begin{equation}
\sqrt{F}\left(  \tau_{ZS},\omega_{ZS}\right)  =\sum_{z}p_{Z}\left(  z\right)
\sqrt{F}\left(  \tau_{S}^{z},\omega_{S}^{z}\right)  ,
\end{equation}
where%
\begin{align}
\tau_{ZS}  &  \equiv\sum_{z}p_{Z}\left(  z\right)  \left\vert z\right\rangle
\left\langle z\right\vert _{Z}\otimes\tau_{S}^{z},\\
\omega_{ZS}  &  \equiv\sum_{z}p_{Z}\left(  z\right)  \left\vert z\right\rangle
\left\langle z\right\vert _{Z}\otimes\omega_{S}^{z}.
\end{align}
The third equality follows from the description of the map $\mathcal{R}%
_{EX_{E}\rightarrow AX_{A}EX_{E}}$ given above. The last inequality is a
consequence of the definition of $F^{\text{sq}}$ because $\rho_{AX_{A}%
BX_{B}EX_{E}}$ is a particular extension of $\rho_{ABX_{B}E}$ and
$\mathcal{R}_{EX_{E}\rightarrow AX_{A}EX_{E}}$ is a particular recovery map.
Given that the chain of inequalities holds for all recovery maps
$\mathcal{R}_{E\rightarrow AE}$ and extensions $\rho_{ABX_{B}E}$ of
$\rho_{ABX_{B}}$, we can conclude that%
\begin{equation}
F^{\operatorname{sq}}\left(  A;BX_{B}\right)  _{\rho}\leq F^{\operatorname{sq}%
}\left(  AX_{A};BX_{B}\right)  _{\rho}.
\end{equation}

\end{proof}

\begin{remark}
The inequalities in Proposition~\ref{prop:inv-cc}\ demonstrate that the
geometric squashed entanglement is monotone non-increasing with respect to
classical communication from Bob to Alice, but not necessarily the other way
around. The essential idea in establishing the inequality
$F^{\operatorname{sq}}\left(  A;BX_{B}\right)  _{\rho}\leq
F^{\operatorname{sq}}\left(  AX_{A};BX_{B}\right)  _{\rho}$ is to give a copy
of the classical data to the party possessing the extension system and to have
the recovery map give a copy to Alice. It is unclear to us whether the other
inequality $F^{\operatorname{sq}}\left(  AX_{A};B\right)  _{\rho}\leq
F^{\operatorname{sq}}\left(  AX_{A};BX_{B}\right)  _{\rho}$ could be
established, given that the recovery operation only goes from an extension
system to Alice, and so it appears that we have no way of giving a copy of
this classical data to Bob.
\end{remark}

The following theorem is a direct consequence of
Corollary~\ref{prop:gse-mono-lo}\ and Proposition~\ref{prop:inv-cc}:

\begin{theorem}
[1-LOCC\ monotone]\label{thm:locc-monotone}The geometric squashed entanglement
is a 1-LOCC monotone, in the sense that it is monotone non-increasing with
respect to local operations and classical communication from Bob to Alice.
\end{theorem}

\begin{theorem}
[Convexity]\label{thm:convex}The geometric squashed entanglement is convex,
i.e.,%
\begin{equation}
\sum_{x}p_{X}( x) E_{F}^{\operatorname{sq}}( A;B) _{\rho^{x}}\geq
E_{F}^{\operatorname{sq}}( A;B) _{\overline{\rho}},
\label{eq:geo-squash-convex}%
\end{equation}
where%
\begin{equation}
\overline{\rho}_{AB}\equiv\sum_{x}p_{X}( x) \rho_{AB}^{x}.
\end{equation}

\end{theorem}

\begin{proof}
Let $\rho_{ABE}^{x}$ be an extension of each $\rho_{AB}^{x}$, so that%
\begin{equation}
\omega_{XABE}\equiv\sum_{x}p_{X}( x) \left\vert x\right\rangle \left\langle
x\right\vert _{X}\otimes\rho_{ABE}^{x}%
\end{equation}
is some extension of $\overline{\rho}_{AB}$. Then the definition of
$E_{F}^{\operatorname{sq}}( A;B) _{\overline{\rho}}$ and
Proposition~\ref{prop:condition-classical} give that%
\begin{align}
2E_{F}^{\operatorname{sq}}( A;B) _{\overline{\rho}}  &  \leq I_{F}\left(
A;B|EX\right)  _{\omega}\\
&  \leq\sum_{x}p_{X}( x) I_{F}( A;B|E) _{\rho^{x}}.
\end{align}
Since the inequality holds independent of each particular extension of
$\rho_{AB}^{x}$, we can conclude (\ref{eq:geo-squash-convex}).
\end{proof}

\begin{theorem}
[Faithfulness]The geometric squashed entanglement is faithful, in the sense
that%
\begin{equation}
E_{F}^{\operatorname{sq}}( A;B) _{\rho}=0~\text{if and only if }\rho
_{AB}\text{ is separable.}%
\end{equation}
This is equivalent to%
\begin{equation}
F^{\operatorname{sq}}( A;B) _{\rho}=1~\text{if and only if }\rho_{AB}\text{ is
separable.}%
\end{equation}
Furthermore, we have the following bound holding for all states $\rho_{AB}$:%
\begin{equation}
E_{F}^{\operatorname{sq}}( A;B) _{\rho}\geq\frac{1}{512\left\vert A\right\vert
^{4}}\left\Vert \rho_{AB}-\operatorname{SEP}(A:B)\right\Vert _{1}^{4}.
\end{equation}

\end{theorem}

\begin{proof}
We first prove the if-part of the theorem. So, given by assumption that
$\rho_{AB}$ is separable, it has a decomposition of the following form:%
\begin{equation}
\rho_{AB}=\sum_{x}p_{X}( x) \left\vert \psi_{x}\right\rangle \left\langle
\psi_{x}\right\vert _{A}\otimes\left\vert \phi_{x}\right\rangle \left\langle
\phi_{x}\right\vert _{B}.
\end{equation}
Then an extension of the state is of the form%
\begin{equation}
\rho_{ABE}=\sum_{x}p_{X}( x) \left\vert \psi_{x}\right\rangle \left\langle
\psi_{x}\right\vert _{A}\otimes\left\vert \phi_{x}\right\rangle \left\langle
\phi_{x}\right\vert _{B}\otimes\left\vert x\right\rangle \left\langle
x\right\vert _{E}.
\end{equation}
Clearly, if the system $A$ becomes lost, someone who possesses system $E$
could measure it and prepare the state $\left\vert \psi_{x}\right\rangle _{A}$
conditioned on the measurement outcome. That is, the recovery map
$\mathcal{R}_{E\rightarrow AE}$ is as follows:%
\begin{equation}
\mathcal{R}_{E\rightarrow AE}\left(  \sigma_{E}\right)  =\sum_{x}\left\langle
x\right\vert \sigma_{E}\left\vert x\right\rangle \ \left\vert \psi
_{x}\right\rangle \left\langle \psi_{x}\right\vert _{A}\otimes\left\vert
x\right\rangle \left\langle x\right\vert _{E}.
\end{equation}
So this implies that%
\begin{equation}
F\left(  \rho_{ABE},\mathcal{R}_{E\rightarrow AE}\left(  \rho_{BE}\right)
\right)  =1,
\end{equation}
and thus $F^{\operatorname{sq}}( A;B) _{\rho}=1$.

The only-if-part of the theorem is a direct consequence of the reasoning
in~\cite{Winterconj}. We repeat the argument from \cite{Winterconj}\ here for
the convenience of the reader. The reasoning from \cite{Winterconj}
establishes that the trace distance between $\rho_{AB}$ and the set
SEP$(A:B)$\ of separable states on systems $A$ and $B$ is bounded from above
by a function of $-1/2\log F^{\operatorname{sq}}( A;B) _{\rho}$ and
$\left\vert A\right\vert $. This will then allow us to conclude the
only-if-part of the theorem.

Let%
\begin{equation}
\varepsilon\equiv-\frac{1}{2}\log F^{\operatorname{sq}}( A;B) _{\rho}
\label{eq:squash-val}%
\end{equation}
for some bipartite state $\rho_{AB}$ and let%
\begin{equation}
\varepsilon_{\omega,\mathcal{R}}\equiv-\frac{1}{2}\log F(\omega_{ABE}%
,\mathcal{R}_{E\rightarrow AE}\left(  \omega_{BE}\right)  ),
\label{eq:particular-squash-val}%
\end{equation}
for some extension $\omega_{ABE}$ and a recovery map $\mathcal{R}%
_{E\rightarrow AE}$. By definition, we have that%
\begin{equation}
\varepsilon=\inf_{\omega,\mathcal{R}_{E\rightarrow AE}}\varepsilon
_{\omega,\mathcal{R}}.
\end{equation}
Then consider that%
\begin{equation}
\varepsilon_{\omega,\mathcal{R}}\geq\frac{1}{8}\left\Vert \omega
_{ABE}-\mathcal{R}_{E\rightarrow AE}\left(  \omega_{BE}\right)  \right\Vert
_{1}^{2}, \label{epsilongsq}%
\end{equation}
where the inequality follows from a well known relation between the fidelity
and trace distance \cite{FG98}. Therefore, by defining $\delta_{\omega
,\mathcal{R}}=\sqrt{8\varepsilon_{\omega,\mathcal{R}}}$ we have that%
\begin{align}
\delta_{\omega,\mathcal{R}}  &  \geq\left\Vert \omega_{ABE}-\mathcal{R}%
_{E\rightarrow AE}\left(  \omega_{BE}\right)  \right\Vert _{1}\\
&  =\left\Vert \omega_{ABE}-\left(  \mathcal{R}_{E\rightarrow A_{2}E}%
\circ\operatorname{Tr}_{A_{1}}\right)  \left(  \omega_{A_{1}BE}\right)
\right\Vert _{1},
\end{align}
where the systems $A_{1}$ and $A_{2}$ are defined to be isomorphic to system
$A$. Now consider applying the same recovery map again. We then have that

\begin{equation}
\delta_{\omega,\mathcal{R}}\geq\left\Vert \left(  \mathcal{R}_{E\rightarrow
A_{3}E}\circ\operatorname{Tr}_{A_{2}}\right)  \left(  \omega_{A_{2}BE}\right)
-\bigcirc_{i=2}^{3}\left(  \mathcal{R}_{E\rightarrow A_{i}E}\circ
\operatorname{Tr}_{A_{i-1}}\right)  \left(  \omega_{A_{1}BE}\right)
\right\Vert _{1},
\end{equation}
which follows from the inequality above and monotonicity of the trace distance
with respect to the quantum operation $\mathcal{R}_{E\rightarrow A_{3}E}\circ
$Tr$_{A_{2}}$. Combining via the triangle inequality, we find for $k\geq2$
that%
\begin{align}
\left\Vert \omega_{ABE}-\bigcirc_{i=2}^{3}\left(  \mathcal{R}_{E\rightarrow
A_{i}E}\circ\operatorname{Tr}_{A_{i-1}}\right)  \left(  \omega_{A_{1}%
BE}\right)  \right\Vert _{1}
\leq2\delta_{\omega,\mathcal{R}}\leq k\delta_{\omega,\mathcal{R}}.
\end{align}
We can iterate this reasoning in the following way:\ For $j\in\left\{
4,\ldots,k\right\}  $ (assuming now $k\geq4$), apply the maps $\mathcal{R}%
_{E\rightarrow A_{j}E}\circ$Tr$_{A_{j-1}}$ along with monotonicity of trace
distance to establish the following inequalities:%
\begin{equation}
\left\Vert \left[  \bigcirc_{i=3}^{j}\left(  \mathcal{R}_{E\rightarrow A_{i}%
E}\circ\operatorname{Tr}_{A_{i-1}}\right)  \left(  \omega_{A_{2}BE}\right)
\right]  -\left[  \bigcirc_{i=2}^{j}\left(  \mathcal{R}_{E\rightarrow A_{i}%
E}\circ\operatorname{Tr}_{A_{i-1}}\right)  \left(  \omega_{A_{1}BE}\right)
\right]  \right\Vert _{1}\leq\delta_{\omega,\mathcal{R}}.
\end{equation}
Apply the triangle inequality to all of these to establish the following
inequalities for $j\in\left\{  1,\ldots,k\right\}  $:%
\begin{equation}
\left\Vert \omega_{ABE}-\bigcirc_{i=2}^{j}\left(  \mathcal{R}_{E\rightarrow
A_{i}E}\circ\operatorname{Tr}_{A_{i-1}}\right)  \left(  \omega_{A_{1}%
BE}\right)  \right\Vert _{1}\leq k\delta_{\omega,\mathcal{R}},
\end{equation}
with the interpretation for $j=1$ that there is no map applied. From
monotonicity of trace distance with respect to quantum operations, we can then
conclude the following inequalities for $j\in\left\{  1,\ldots,k\right\}  $:%
\begin{equation}
\left\Vert \rho_{AB}-\operatorname{Tr}_{E}\left\{  \bigcirc_{i=2}^{j}\left(
\mathcal{R}_{E\rightarrow A_{i}E}\circ\operatorname{Tr}_{A_{i-1}}\right)
\left(  \omega_{A_{1}BE}\right)  \right\}  \right\Vert _{1}\leq k\delta
_{\omega,\mathcal{R}}.\label{eq:faithful-ineqs-1}%
\end{equation}
Let $\gamma_{A_{1}A_{2}\cdots A_{k}BE}$\ denote the following state:%
\begin{equation}
\gamma_{A_{1}A_{2}\cdots A_{k}BE}\equiv\mathcal{R}_{E\rightarrow A_{k}%
E}\left(  \cdots\left(  \mathcal{R}_{E\rightarrow A_{2}E}\left(  \omega
_{A_{1}BE}\right)  \right)  \right)  .
\end{equation}
(See Figure~\ref{fig:faithfulness} for a graphical depiction of this
state.)\begin{figure}[ptb]
\center
\includegraphics[scale=1.0]{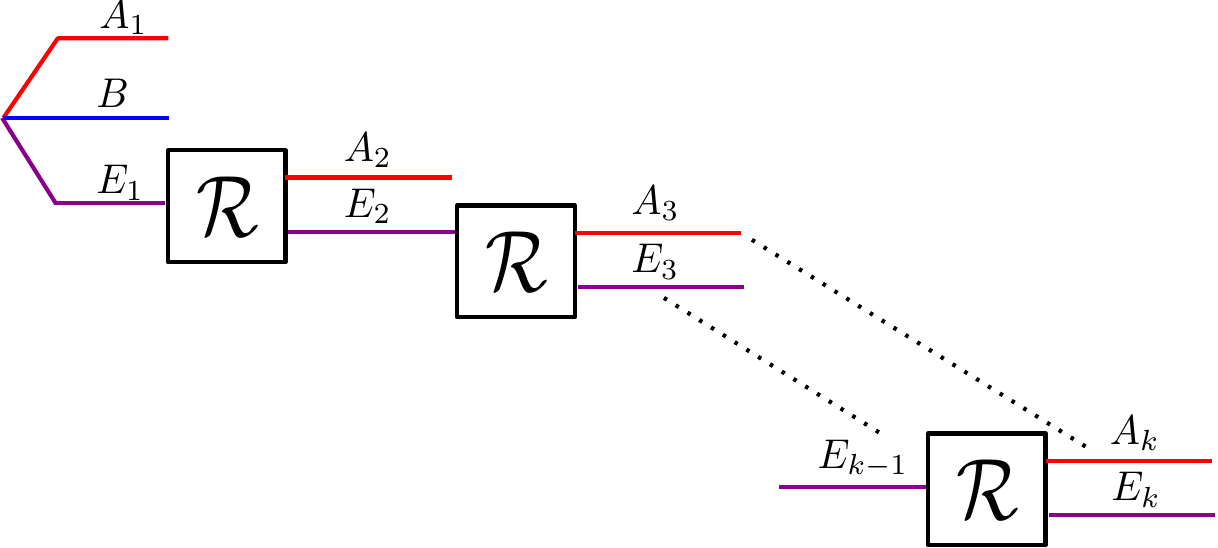}\caption{This figure
illustrates the global state after performing a recovery map $k$ times on
system~$E$.}%
\label{fig:faithfulness}%
\end{figure}Then the inequalities in (\ref{eq:faithful-ineqs-1}) are
equivalent to the following inequalities for $j\in\left\{  1\ldots,k\right\}
$:%
\begin{equation}
\left\Vert \rho_{AB}-\gamma_{A_{j}B}\right\Vert _{1}\leq k\delta
_{\omega,\mathcal{R}},
\end{equation}
which are in turn equivalent to the following ones for any permutation $\pi\in
S_{k}$:%
\begin{equation}
\left\Vert \rho_{AB}-\operatorname{Tr}_{A_{2}\cdots A_{k}}\left\{
W_{A_{1}A_{2}\cdots A_{k}}^{\pi}\gamma_{A_{1}A_{2}\cdots A_{k}B}\left(
W_{A_{1}A_{2}\cdots A_{k}}^{\pi}\right)  ^{\dag}\right\}  \right\Vert _{1}\leq
k\delta_{\omega,\mathcal{R}},\label{eq:faithful-ineqs-2}%
\end{equation}
with $W_{A_{1}A_{2}\cdots A_{k}}^{\pi}$ a unitary representation of the
permutation $\pi$. We can then define $\overline{\gamma}_{A_{1}\cdots A_{k}B}$
as a symmetrized version of $\gamma_{A_{1}\cdots A_{k}B}$:%
\begin{equation}
\overline{\gamma}_{A_{1}\cdots A_{k}B}\equiv\frac{1}{k!}\sum_{\pi\in S_{k}%
}W_{A_{1}A_{2}\cdots A_{k}}^{\pi}\gamma_{A_{1}\cdots A_{k}B}\left(
W_{A_{1}A_{2}\cdots A_{k}}^{\pi}\right)  ^{\dag}.
\end{equation}
The inequalities in (\ref{eq:faithful-ineqs-2}) allow us to conclude that%
\begin{align}
k\delta_{\omega,\mathcal{R}} &  \geq\frac{1}{k!}\sum_{\pi\in S_{k}}\left\Vert
\rho_{AB}-\operatorname{Tr}_{A_{2}\cdots A_{k}}\left\{  W_{A_{1}A_{2}\cdots
A_{k}}^{\pi}\gamma_{A_{1}A_{2}\cdots A_{k}B}\left(  W_{A_{1}A_{2}\cdots A_{k}%
}^{\pi}\right)  ^{\dag}\right\}  \right\Vert _{1}\\
&  \geq\left\Vert \rho_{AB}-\operatorname{Tr}_{A_{2}\cdots A_{k}}\left\{
\frac{1}{k!}\sum_{\pi\in S_{k}}W_{A_{1}A_{2}\cdots A_{k}}^{\pi}\gamma
_{A_{1}A_{2}\cdots A_{k}B}\left(  W_{A_{1}A_{2}\cdots A_{k}}^{\pi}\right)
^{\dag}\right\}  \right\Vert _{1}\\
&  =\left\Vert \rho_{AB}-\overline{\gamma}_{A_{1}B}\right\Vert _{1}%
,\label{eq:faithful-ineqs-3}%
\end{align}

where the second inequality is a consequence of the convexity of trace
distance. So what the reasoning in \cite{Winterconj} accomplishes is to
construct a $k$-extendible state $\overline{\gamma}_{A_{1}B}$ that is
$k\delta_{\omega,\mathcal{R}}$-close to $\rho_{AB}$ in trace distance.

Following \cite{Winterconj}, we now recall a particular quantum de Finetti
result in~\cite[Theorem II.7']{CKMR07}. Consider a state $\omega_{A_{1}\cdots
A_{k}B}$ which is permutation invariant with respect to systems $A_{1}\cdots
A_{k}$. Let $\omega_{A_{1}\cdots A_{n}B}$ denote the reduced state on $n$ of
the $k$ $A$ systems where $n\leq k$. Then, for large $k$, $\omega_{A_{1}\cdots
A_{n}B}$ is close in trace distance to a convex combination of product states
of the form $\int\sigma_{A}^{\otimes n}\otimes\tau\left(  \sigma\right)
_{B}\ d\mu(\sigma)$, where $\mu$ is a probability measure on the set of mixed
states on a single $A$ system and $\left\{  \tau\left(  \sigma\right)
\right\}  _{\sigma}$ is a family of states parametrized by $\sigma$, with the
approximation given by%
\begin{equation}
\frac{2\left\vert A\right\vert ^{2}n}{k}\geq\left\Vert \omega_{A_{1}\cdots
A_{n}B}-\int\sigma_{A}^{\otimes n}\otimes\tau\left(  \sigma\right)  _{B}%
\ d\mu(\sigma)\right\Vert _{1}.
\end{equation}
Applying this theorem in our context (choosing $n=1$) leads to the following
conclusion:%
\begin{align}
\frac{2\left\vert A\right\vert ^{2}}{k}  &  \geq\left\Vert \overline{\gamma
}_{A_{1}B}-\int\sigma_{A_{1}}\otimes\tau\left(  \sigma\right)  _{B}%
\ d\mu(\sigma)\right\Vert _{1}\\
&  \geq\left\Vert \overline{\gamma}_{A_{1}B}-\text{SEP}(A_{1}:B)\right\Vert
_{1}, \label{eq:dist-to-sep}%
\end{align}
because the state $\int\sigma_{A_{1}}\otimes\tau\left(  \sigma\right)
_{B}\ d\mu(\sigma)$ is a particular separable state.

We can now combine (\ref{eq:faithful-ineqs-3}) and (\ref{eq:dist-to-sep}%
)\ with the triangle inequality to conclude the following bound%
\begin{equation}
\left\Vert \rho_{AB}-\text{SEP}(A:B)\right\Vert _{1}\leq\frac{2|A|^{2}}%
{k}+k\delta_{\omega,\mathcal{R}}.
\end{equation}
By choosing $k$ to diverge slower than $\delta_{\omega,\mathcal{R}}^{-1}$, say
as $k=|A|\sqrt{2/\delta_{\omega,\mathcal{R}}}$, we obtain the following bound:%
\begin{align}
  \left\Vert \rho_{AB}-\text{SEP}(A:B)\right\Vert _{1}
&  \leq\left\vert A\right\vert \sqrt{8\delta_{\omega,\mathcal{R}}}\\
&  =\left(  512\right)  ^{1/4}\left\vert A\right\vert \varepsilon
_{\omega,\mathcal{R}}^{1/4}.
\end{align}
Since the above bound holds for all extensions and recovery maps, we can
obtain the tightest bound by taking an infimum over all of them. By
substituting with (\ref{eq:squash-val}) and (\ref{eq:particular-squash-val}),
we find that%
\begin{equation}
\left\Vert \rho_{AB}-\text{SEP}(A:B)\right\Vert _{1}\leq
\left(  512\right)  ^{1/4}\left\vert A\right\vert \left(  -1/2\log
F^{\operatorname{sq}}( A;B) _{\rho}\right)  ^{1/4},
\end{equation}
or equivalently%
\begin{align}
E_{F}^{\operatorname{sq}}( A;B) _{\rho}  &  =-1/2\log F^{\operatorname{sq}}(
A;B) _{\rho}\\
&  \geq\frac{1}{512\left\vert A\right\vert ^{4}}\left\Vert \rho_{AB}%
-\text{SEP}(A:B)\right\Vert _{1}^{4}.
\end{align}
This proves the converse part of the faithfulness of the geometric squashed entanglement.
\end{proof}

\begin{proposition}
[Reduction to geometric measure]\label{prop:pure-state}Let $\phi_{AB}$ be a
bipartite pure state. Then%
\begin{align}
E_{F}^{\operatorname{sq}}( A;B) _{\phi}  &  =-\frac{1}{2}\log\sup_{\left\vert
\varphi\right\rangle _{A}}\left\langle \phi\right\vert _{AB}\left(
\varphi_{A}\otimes\phi_{B}\right)  \left\vert \phi\right\rangle _{AB}%
\label{eq:geo-reduce}\\
&  =-\log\left\Vert \phi_{A}\right\Vert _{\infty}. \label{eq:geo-reduce-2}%
\end{align}

\end{proposition}

\begin{proof}
Any extension of a pure bipartite state is of the form $\phi_{AB}\otimes
\omega_{E}$, where $\omega_{E}$ is some state. Applying
Proposition~\ref{prop:condition-on-product}, we find that%
\begin{align}
F( A;B|E) _{\phi\otimes\omega}  &  =F( A;B) _{\phi}\\
&  =\sup_{\sigma_{A}}F\left(  \phi_{AB},\sigma_{A}\otimes\phi_{B}\right) \\
&  =\sup_{\left\vert \varphi\right\rangle _{A}}\left\langle \phi\right\vert
_{AB}\left(  \varphi_{A}\otimes\phi_{B}\right)  \left\vert \phi\right\rangle
_{AB}.
\end{align}
The last equality follows due to a convexity argument applied to%
\begin{equation}
F\left(  \phi_{AB},\sigma_{A}\otimes\phi_{B}\right)  =\left\langle
\phi\right\vert _{AB}\sigma_{A}\otimes\phi_{B}\left\vert \phi\right\rangle
_{AB}.
\end{equation}
Since the equality holds independent of any particular extension of $\phi
_{AB}$, we obtain (\ref{eq:geo-reduce})\ upon applying a negative logarithm
and dividing by two. The other equality (\ref{eq:geo-reduce-2})\ follows
because%
\begin{align}
  \left\langle \phi\right\vert _{AB}\left(  \varphi_{A}\otimes\phi
_{B}\right)  \left\vert \phi\right\rangle _{AB}
&  =\left\langle \phi\right\vert _{AB}\left(  \varphi_{A}\phi_{A}\otimes
I_{B}\right)  \left\vert \phi\right\rangle _{AB}\\
&  =\text{Tr}\left\{  \left\vert \phi\right\rangle \left\langle \phi
\right\vert _{AB}\left(  \varphi_{A}\phi_{A}\otimes I_{B}\right)  \right\} \\
&  =\text{Tr}\left\{  \phi_{A}\varphi_{A}\phi_{A}\right\} \\
&  =\left\langle \varphi\right\vert _{A}\phi_{A}^{2}\left\vert \varphi
\right\rangle _{A}.
\end{align}
Taking a supremum over all unit vectors $\left\vert \varphi\right\rangle _{A}$
then gives%
\begin{equation}
E_{F}^{\operatorname{sq}}( A;B) _{\phi}=-\frac{1}{2}\log\left\Vert \phi
_{A}^{2}\right\Vert _{\infty},
\end{equation}
which is equivalent to (\ref{eq:geo-reduce-2}).
\end{proof}

\begin{proposition}
[Normalization]\label{prop:normalization}For a maximally entangled state
$\Phi_{AB}$ of Schmidt rank $d$,%
\begin{equation}
E_{F}^{\operatorname{sq}}( A;B) _{\Phi}=\log d.
\end{equation}

\end{proposition}

\begin{proof}
This follows directly from (\ref{eq:geo-reduce-2}) of
Proposition~\ref{prop:pure-state} because $\Phi_{A}=I_{A}/d$.
\end{proof}

\begin{proposition}
For a private state $\gamma_{ABA^{\prime}B^{\prime}}$\ of $\log d$ private
bits, the geometric squashed entanglement obeys the following bound:%
\begin{equation}
E_{F}^{\operatorname{sq}}\left(  AA^{\prime};BB^{\prime}\right)  _{\gamma}%
\geq\log d.
\end{equation}

\end{proposition}

\begin{proof}
The proof is in a similar spirit to the proof of \cite[Proposition 4.19]{C06},
but tailored to the fidelity of recovery quantity. Recall (\ref{eq:private-1}%
)-(\ref{eq:private-last}). Any extension $\gamma_{ABA^{\prime}B^{\prime}E}%
$\ of a private state $\gamma_{ABA^{\prime}B^{\prime}}$ takes the form:%
\begin{equation}
\gamma_{ABA^{\prime}B^{\prime}E}=U_{ABA^{\prime}B^{\prime}}\left(  \Phi
_{AB}\otimes\rho_{A^{\prime}B^{\prime}E}\right)  U_{ABA^{\prime}B^{\prime}%
}^{\dag},
\end{equation}
where $\rho_{A^{\prime}B^{\prime}E}$ is an extension of $\rho_{A^{\prime
}B^{\prime}}$. This is because the state $\Phi_{AB}$ is not extendible. Then
consider that%
\begin{equation}
F\left(  AA^{\prime};BB^{\prime}|E\right)  _{\gamma}=\sup_{\mathcal{R}%
}F\left(  \gamma_{ABA^{\prime}B^{\prime}E},\mathcal{R}_{E\rightarrow
AA^{\prime}E}\left(  \gamma_{BB^{\prime}E}\right)  \right)  ,
\end{equation}
where $\mathcal{R}_{E\rightarrow AA^{\prime}E}$ is a recovery map. From
(\ref{eq:private-1})-(\ref{eq:private-last}), we can write%
\begin{equation}
\gamma_{ABA^{\prime}B^{\prime}E}=\frac{1}{d}\sum_{i,j}\left\vert
i\right\rangle \left\langle j\right\vert _{A}\otimes\left\vert i\right\rangle
\left\langle j\right\vert _{B}\otimes V_{A^{\prime}B^{\prime}}^{i}%
\rho_{A^{\prime}B^{\prime}E}\left(  V_{A^{\prime}B^{\prime}}^{j}\right)
^{\dag},
\end{equation}
which implies that%
\begin{equation}
\gamma_{BB^{\prime}E}=\frac{1}{d}\sum_{i}\left\vert i\right\rangle
\left\langle i\right\vert _{B}\otimes\operatorname{Tr}_{\hat{A}^{\prime}%
}\left\{  V_{\hat{A}^{\prime}B^{\prime}}^{i}\rho_{\hat{A}^{\prime}B^{\prime}%
E}\left(  V_{\hat{A}^{\prime}B^{\prime}}^{i}\right)  ^{\dag}\right\}  .
\end{equation}
So then consider the fidelity of recovery for a particular recovery map
$\mathcal{R}_{E\rightarrow AA^{\prime}E}$: 
\begin{align}
&  F\left(  \gamma_{ABA^{\prime}B^{\prime}E},\mathcal{R}_{E\rightarrow
AA^{\prime}E}\left(  \gamma_{BB^{\prime}E}\right)  \right)  \nonumber\\
&  =F\left(  U_{ABA^{\prime}B^{\prime}}\left(  \Phi_{AB}\otimes\rho
_{A^{\prime}B^{\prime}E}\right)  U_{ABA^{\prime}B^{\prime}}^{\dag},\frac{1}%
{d}\sum_{i}\left\vert i\right\rangle \left\langle i\right\vert _{B}%
\otimes\mathcal{R}_{E\rightarrow AA^{\prime}E}\left(  \operatorname{Tr}%
_{\hat{A}^{\prime}}\left\{  V_{\hat{A}^{\prime}B^{\prime}}^{i}\rho_{\hat
{A}^{\prime}B^{\prime}E}\left(  V_{\hat{A}^{\prime}B^{\prime}}^{i}\right)
^{\dag}\right\}  \right)  \right)  \\
&  =F\left(  \left(  \Phi_{AB}\otimes\rho_{A^{\prime}B^{\prime}E}\right)
,U_{ABA^{\prime}B^{\prime}}^{\dag}\left[  \frac{1}{d}\sum_{i}\left\vert
i\right\rangle \left\langle i\right\vert _{B}\otimes\mathcal{R}_{E\rightarrow
AA^{\prime}E}\left(  \operatorname{Tr}_{\hat{A}^{\prime}}\left\{  V_{\hat
{A}^{\prime}B^{\prime}}^{i}\rho_{\hat{A}^{\prime}B^{\prime}E}\left(
V_{\hat{A}^{\prime}B^{\prime}}^{i}\right)  ^{\dag}\right\}  \right)  \right]
U_{ABA^{\prime}B^{\prime}}\right)  ,\label{eq:private-state-bound}%
\end{align}
where the second equality follows from invariance of the fidelity with respect
to unitaries. Then consider that%
\begin{align}
&  U_{ABA^{\prime}B^{\prime}}^{\dag}\left[  \frac{1}{d}\sum_{i}\left\vert
i\right\rangle \left\langle i\right\vert _{B}\otimes\mathcal{R}_{E\rightarrow
AA^{\prime}E}\left(  \operatorname{Tr}_{\hat{A}^{\prime}}\left\{  V_{\hat
{A}^{\prime}B^{\prime}}^{i}\rho_{\hat{A}^{\prime}B^{\prime}E}\left(
V_{\hat{A}^{\prime}B^{\prime}}^{i}\right)  ^{\dag}\right\}  \right)  \right]
U_{ABA^{\prime}B^{\prime}}\nonumber\\
&  =\left(  I_{A}\otimes\sum_{j}\left\vert j\right\rangle \left\langle
j\right\vert _{B}\otimes\left(  V_{A^{\prime}B^{\prime}}^{j}\right)  ^{\dag
}\right)  \left[  \frac{1}{d}\sum_{i}\left\vert i\right\rangle \left\langle
i\right\vert _{B}\otimes\mathcal{R}_{E\rightarrow AA^{\prime}E}\left(
\operatorname{Tr}_{\hat{A}^{\prime}}\left\{  V_{\hat{A}^{\prime}B^{\prime}%
}^{i}\rho_{\hat{A}^{\prime}B^{\prime}E}\left(  V_{\hat{A}^{\prime}B^{\prime}%
}^{i}\right)  ^{\dag}\right\}  \right)  \right]  \times\nonumber\\
&  \ \ \ \ \ \ \ \ \ \ \ \ \left(  I_{A}\otimes\sum_{j^{\prime}}\left\vert
j^{\prime}\right\rangle \left\langle j^{\prime}\right\vert _{B}\otimes
V_{A^{\prime}B^{\prime}}^{j^{\prime}}\right)  \\
&  =\frac{1}{d}\sum_{i}\left\vert i\right\rangle \left\langle i\right\vert
_{B}\otimes\left(  V_{A^{\prime}B^{\prime}}^{i}\right)  ^{\dag}\mathcal{R}%
_{E\rightarrow AA^{\prime}E}\left(  \operatorname{Tr}_{\hat{A}^{\prime}%
}\left\{  V_{\hat{A}^{\prime}B^{\prime}}^{i}\rho_{\hat{A}^{\prime}B^{\prime}%
E}\left(  V_{\hat{A}^{\prime}B^{\prime}}^{i}\right)  ^{\dag}\right\}  \right)
V_{A^{\prime}B^{\prime}}^{i}.
\end{align}
If we trace over systems $A^{\prime}B^{\prime}$, the fidelity only goes up, so
consider that the state above becomes as follows after taking this partial
trace:%
\begin{align}
&  \frac{1}{d}\sum_{i}\left\vert i\right\rangle \left\langle i\right\vert
_{B}\otimes\operatorname{Tr}_{A^{\prime}B^{\prime}}\left\{  \left(
V_{A^{\prime}B^{\prime}}^{i}\right)  ^{\dag}\mathcal{R}_{E\rightarrow
AA^{\prime}E}\left(  \operatorname{Tr}_{\hat{A}^{\prime}}\left\{  V_{\hat
{A}^{\prime}B^{\prime}}^{i}\rho_{\hat{A}^{\prime}B^{\prime}E}\left(
V_{\hat{A}^{\prime}B^{\prime}}^{i}\right)  ^{\dag}\right\}  \right)
V_{A^{\prime}B^{\prime}}^{i}\right\}  \nonumber\\
&  =\frac{1}{d}\sum_{i}\left\vert i\right\rangle \left\langle i\right\vert
_{B}\otimes\operatorname{Tr}_{A^{\prime}B^{\prime}}\left\{  \mathcal{R}%
_{E\rightarrow AA^{\prime}E}\left(  \operatorname{Tr}_{\hat{A}^{\prime}%
}\left\{  V_{\hat{A}^{\prime}B^{\prime}}^{i}\rho_{\hat{A}^{\prime}B^{\prime}%
E}\left(  V_{\hat{A}^{\prime}B^{\prime}}^{i}\right)  ^{\dag}\right\}  \right)
\right\}  \\
&  =\frac{1}{d}\sum_{i}\left\vert i\right\rangle \left\langle i\right\vert
_{B}\otimes\operatorname{Tr}_{A^{\prime}}\left\{  \mathcal{R}_{E\rightarrow
AA^{\prime}E}\left(  \operatorname{Tr}_{\hat{A}^{\prime}B^{\prime}}\left\{
V_{\hat{A}^{\prime}B^{\prime}}^{i}\rho_{\hat{A}^{\prime}B^{\prime}E}\left(
V_{\hat{A}^{\prime}B^{\prime}}^{i}\right)  ^{\dag}\right\}  \right)  \right\}
\\
&  =\frac{1}{d}\sum_{i}\left\vert i\right\rangle \left\langle i\right\vert
_{B}\otimes\operatorname{Tr}_{A^{\prime}}\left\{  \mathcal{R}_{E\rightarrow
AA^{\prime}E}\left(  \operatorname{Tr}_{\hat{A}^{\prime}B^{\prime}}\left\{
\rho_{\hat{A}^{\prime}B^{\prime}E}\right\}  \right)  \right\}  \\
&  =\frac{1}{d}\sum_{i}\left\vert i\right\rangle \left\langle i\right\vert
_{B}\otimes\operatorname{Tr}_{A^{\prime}}\left\{  \mathcal{R}_{E\rightarrow
AA^{\prime}E}\left(  \rho_{E}\right)  \right\}  \\
&  =\pi_{B}\otimes\mathcal{R}_{E\rightarrow AE}\left(  \rho_{E}\right)  ,
\end{align}

where $\pi_{B}$ is a maximally mixed state on system $B$. So an upper bound on
(\ref{eq:private-state-bound}) is given by%
\begin{align}
  F\left(  \Phi_{AB}\otimes\rho_{E},\pi_{B}\otimes\mathcal{R}_{E\rightarrow
AE}\left(  \rho_{E}\right)  \right) 
&  \leq F\left(  \Phi_{AB},\pi_{B}\otimes\mathcal{R}_{E\rightarrow A}\left(
\rho_{E}\right)  \right) \\
&  =1/d^{2}.
\end{align}
Since this upper bound is universal for any recovery map and any extension of
the original state, we obtain the following inequality:%
\begin{equation}
\sup_{\substack{\gamma_{ABA^{\prime}B^{\prime}E}:\\\gamma_{ABA^{\prime
}B^{\prime}}=\operatorname{Tr}_{E}\left\{  \gamma_{ABA^{\prime}B^{\prime}%
E}\right\}  }}F\left(  AA^{\prime};BB^{\prime}|E\right)  _{\gamma}\leq1/d^{2}.
\end{equation}
After taking a negative logarithm, we recover the statement of the proposition.
\end{proof}

\begin{proposition}
[Subadditivity]Let $\omega_{A_{1}B_{1}A_{2}B_{2}}\equiv\rho_{A_{1}B_{1}%
}\otimes\sigma_{A_{2}B_{2}}$. Then%
\begin{equation}
E_{F}^{\operatorname{sq}}\left(  A_{1}A_{2};B_{1}B_{2}\right)  _{\omega}\leq
E_{F}^{\operatorname{sq}}\left(  A_{1};B_{1}\right)  _{\rho}+E_{F}%
^{\operatorname{sq}}\left(  A_{2};B_{2}\right)  _{\sigma},
\end{equation}
which is equivalent to%
\begin{equation}
F^{\operatorname{sq}}\left(  A_{1};B_{1}\right)  _{\rho}\cdot
F^{\operatorname{sq}}\left(  A_{2};B_{2}\right)  _{\tau}\leq
F^{\operatorname{sq}}\left(  A_{1}A_{2};B_{1}B_{2}\right)  _{\rho\otimes\tau}.
\end{equation}

\end{proposition}

\begin{proof}
Let $\rho_{A_{1}B_{1}E_{1}}$ be an extension of $\rho_{A_{1}B_{1}}$ and let
$\tau_{A_{2}B_{2}E_{2}}$ be an extension of $\tau_{A_{2}B_{2}}$. Let
$\mathcal{R}_{E_{1}\rightarrow A_{1}E_{1}}^{1}$ and $\mathcal{R}%
_{E_{2}\rightarrow A_{2}E_{2}}^{2}$ be recovery maps.\ Then
\begin{align}
&  F\left(  \rho_{A_{1}B_{1}E_{1}},\mathcal{R}_{E_{1}\rightarrow A_{1}E_{1}%
}^{1}\left(  \rho_{B_{1}E_{1}}\right)  \right)  \cdot F\left(  \tau
_{A_{2}B_{2}E_{2}},\mathcal{R}_{E_{2}\rightarrow A_{2}E_{2}}^{2}\left(
\tau_{B_{2}E_{2}}\right)  \right) \nonumber\\
&  =F\left(  \rho_{A_{1}B_{1}E_{1}}\otimes\tau_{A_{2}B_{2}E_{2}}%
,\mathcal{R}_{E_{1}\rightarrow A_{1}E_{1}}^{1}\left(  \rho_{B_{1}E_{1}%
}\right)  \otimes\mathcal{R}_{E_{2}\rightarrow A_{2}E_{2}}^{2}\left(
\tau_{B_{2}E_{2}}\right)  \right) \\
&  \leq\sup_{\omega_{A_{1}A_{2}B_{1}B_{2}E}}\sup_{\mathcal{R}_{E\rightarrow
A_{1}A_{2}E}}\left\{  F\left(  \omega_{A_{1}A_{2}B_{1}B_{2}E},\mathcal{R}%
_{E\rightarrow A_{1}A_{2}E}\left(  \omega_{B_{1}B_{2}E}\right)  \right)
:\rho_{A_{1}B_{1}}\otimes\tau_{A_{2}B_{2}}=\operatorname{Tr}_{E}\left\{
\omega_{A_{1}A_{2}B_{1}B_{2}E}\right\}  \right\} \\
&  =F^{\operatorname{sq}}\left(  A_{1}A_{2};B_{1}B_{2}\right)  _{\rho
\otimes\tau}.
\end{align}

Since the inequality holds for all extensions $\rho_{A_{1}B_{1}E_{1}}$ and
$\tau_{A_{2}B_{2}E_{2}}$ and recovery maps $\mathcal{R}_{E_{1}\rightarrow
A_{1}E_{1}}^{1}$ and $\mathcal{R}_{E_{2}\rightarrow A_{2}E_{2}}^{2}$, we can
conclude that%
\begin{equation}
F^{\operatorname{sq}}\left(  A_{1};B_{1}\right)  _{\rho}\cdot
F^{\operatorname{sq}}\left(  A_{2};B_{2}\right)  _{\tau}\leq
F^{\operatorname{sq}}\left(  A_{1}A_{2};B_{1}B_{2}\right)  _{\rho\otimes\tau}.
\end{equation}
By taking negative logarithms and dividing by 1/2, we arrive at the
subadditivity statement for $E_{F}^{\operatorname{sq}}$.
\end{proof}

\begin{proposition}
[Continuity]\label{prop:geo-SE-cont}The geometric squashed entanglement is a
continuous function of its input. That is, given two bipartite states
$\rho_{AB}$ and $\sigma_{AB}$ such that $F\left(  \rho_{AB},\sigma
_{AB}\right)  \geq1-\varepsilon$ where $\varepsilon\in\left[  0,1\right]  $,
then the following inequalities hold%
\begin{align}
\left\vert F^{\operatorname{sq}}( A;B) _{\rho}-F^{\operatorname{sq}}( A;B)
_{\sigma}\right\vert  &  \leq8\sqrt{\varepsilon},\label{eq:f-sq-cont}\\
\left\vert E_{F}^{\operatorname{sq}}( A;B) _{\rho}-E_{F}^{\operatorname{sq}}(
A;B) _{\sigma}\right\vert  &  \leq4\left\vert A\right\vert ^{2}\sqrt
{\varepsilon}. \label{eq:e-f-sq-cont}%
\end{align}

\end{proposition}

\begin{proof}
This is a direct consequence of the continuity of fidelity of recovery
(Proposition~\ref{prop:FoR-continuity}). Letting $\sigma_{ABE}$ be an
arbitrary extension of $\sigma_{AB}$, \cite[Corollary 9]{TCR10}\ implies that
there exists an extension $\rho_{ABE}$ of $\rho_{AB}$ such that%
\begin{equation}
F\left(  \rho_{ABE},\sigma_{ABE}\right)  \geq1-\varepsilon.
\end{equation}
By Proposition~\ref{prop:FoR-continuity}, we can conclude that%
\begin{align}
F( A;B|E) _{\sigma}  &  \leq F( A;B|E) _{\rho}+8\sqrt{\varepsilon}\\
&  \leq F^{\operatorname{sq}}( A;B) _{\rho}+8\sqrt{\varepsilon}.
\end{align}
Given that the extension of $\sigma_{AB}$ is arbitrary, we can conclude that%
\begin{equation}
F^{\operatorname{sq}}( A;B) _{\sigma}\leq F^{\operatorname{sq}}( A;B) _{\rho
}+8\sqrt{\varepsilon}.
\end{equation}
A similar argument gives that%
\begin{equation}
F^{\operatorname{sq}}( A;B) _{\rho}\leq F^{\operatorname{sq}}( A;B) _{\sigma
}+8\sqrt{\varepsilon},
\end{equation}
from which we can conclude (\ref{eq:f-sq-cont}). We then obtain
(\ref{eq:e-f-sq-cont}) by the same line of reasoning that led us to
(\ref{eq:I-FoR-continuity}).
\end{proof}

\section{Fidelity of recovery from a quantum measurement}

\label{sec:FoMR}In this section, we propose an alternative measure of quantum
correlations, the \textit{surprisal of measurement recoverability}, which
follows the original motivation behind the quantum discord \cite{zurek}.
However, our measure has a clear operational meaning in the \textquotedblleft
one-shot\textquotedblright\ setting, being based on how well one can recover a
bipartite quantum state if one system is measured. We begin by recalling the
definition of the quantum discord and proceed from there with the motivation
behind the newly proposed measure.

\begin{definition}
[Quantum discord]The quantum discord of a bipartite state $\rho_{AB}$ is
defined as the difference between the quantum mutual information of $\rho
_{AB}$ and the classical correlation \cite{HV01}\ of $\rho_{AB}$:%
\begin{align}
D( \overline{A};B) _{\rho}  &  \equiv I( A;B) _{\rho}-\sup_{\left\{
\Lambda^{x}\right\}  }I\left(  X;B\right)  _{\sigma}\\
&  =\inf_{\left\{  \Lambda^{x}\right\}  }\left[  I( A;B) _{\rho}-I\left(
X;B\right)  _{\sigma}\right]  , \label{eq:disc-obj-func}%
\end{align}
where $\left\{  \Lambda^{x}\right\}  $ is a POVM\ with $\Lambda^{x}\geq0$ for
all $x$ and $\sum_{x}\Lambda^{x}=I$ and $\sigma_{XB}$ is defined as%
\begin{equation}
\sigma_{XB}\equiv\sum_{x}\left\vert x\right\rangle \left\langle x\right\vert
_{X}\otimes\operatorname{Tr}_{A}\left\{  \Lambda_{A}^{x}\rho_{AB}\right\}  .
\label{eq:sigma-state}%
\end{equation}

\end{definition}

We now recall how to write the quantum discord in terms of conditional mutual
information as done explicitly in \cite{P12} (see also \cite{BSW14} and
\cite{SBW14}). Let $\mathcal{M}_{A\rightarrow X}$ denote the following
measurement map:%
\begin{equation}
\mathcal{M}_{A\rightarrow X}\left(  \omega_{A}\right)  \equiv\sum
_{x}\operatorname{Tr}\left\{  \Lambda_{A}^{x}\omega_{A}\right\}  \left\vert
x\right\rangle \left\langle x\right\vert _{X}. \label{eq:meas-map}%
\end{equation}
Using this, we can write \eqref{eq:sigma-state} as$\ \sigma_{XB}%
=\mathcal{M}_{A\rightarrow X}( \rho_{AB}) $. Now, to every measurement map
$\mathcal{M}_{A\rightarrow X}$, we can find an isometric extension of it,
having the following form:%
\begin{equation}
U_{A\rightarrow XE}^{\mathcal{M}}\left\vert \psi\right\rangle _{A}\equiv
\sum_{x}\left\vert x\right\rangle _{X}\left\vert x,y\right\rangle
_{E}\left\langle \varphi_{x,y}\right\vert _{A}\left\vert \psi\right\rangle
_{A}, \label{eq:meas-isometry}%
\end{equation}
where the vectors $\left\{  \left\vert \varphi_{x,y}\right\rangle
_{A}\right\}  $ are part of a rank-one refinement of the POVM $\left\{
\Lambda_{A}^{x}\right\}  $:%
\begin{equation}
\Lambda_{A}^{x}=\sum_{y}\left\vert \varphi_{x,y}\right\rangle \left\langle
\varphi_{x,y}\right\vert .
\end{equation}
(In the above, we are taking a spectral decomposition of the operator
$\Lambda_{A}^{x}$.) Thus,%
\begin{equation}
\mathcal{M}_{A\rightarrow X}\left(  \omega_{A}\right)  =\operatorname{Tr}%
_{E}\left\{  \mathcal{U}_{A\rightarrow XE}^{\mathcal{M}}\left(  \omega
_{A}\right)  \right\}  ,
\end{equation}
where%
\begin{equation}
\mathcal{U}_{A\rightarrow XE}^{\mathcal{M}}\left(  \omega_{A}\right)  \equiv
U_{A\rightarrow XE}^{\mathcal{M}}\left(  \omega_{A}\right)  \left(
U_{A\rightarrow XE}^{\mathcal{M}}\right)  ^{\dag}. \label{eq:iso-meas-map}%
\end{equation}
Let $\sigma_{XEB}$ denote the following state:%
\begin{equation}
\sigma_{XEB}=\mathcal{U}_{A\rightarrow XE}^{\mathcal{M}}\left(  \rho
_{AB}\right)  .
\end{equation}

We can use the above development to rewrite the objective function of the
quantum discord in (\ref{eq:disc-obj-func}) as follows:%
\begin{align}
I( A;B) _{\rho}-I\left(  X;B\right)  _{\sigma}  &  =I\left(  XE;B\right)
_{\sigma}-I\left(  X;B\right)  _{\sigma}\\
&  =I\left(  E;B|X\right)  _{\sigma}.
\end{align}
So this means that we can rewrite the discord in terms of the conditional
mutual information as%
\begin{equation}
D( \overline{A};B) =\inf_{\left\{  \Lambda^{x}\right\}  }I\left(
E;B|X\right)  _{\sigma}, \label{eq:discord-CMI}%
\end{equation}
with the state $\sigma_{XEB}$ understood as described above, as arising from
an isometric extension of a measurement map applied to the state $\rho_{AB}$.
We are now in a position to define the surprisal of measurement recoverability:

\begin{definition}
[Surprisal of meas. recoverability]We define the following information
quantity:%
\begin{equation}
D_{F}( \overline{A};B) _{\rho}\equiv\inf_{\left\{  \Lambda^{x}\right\}  }%
I_{F}\left(  E;B|X\right)  _{\sigma}, \label{eq:discord-F-CMI}%
\end{equation}
where we have simply substituted the conditional mutual information in
\eqref{eq:discord-CMI}\ with $I_{F}$. Writing out the right-hand side of
\eqref{eq:discord-F-CMI}\ carefully, we find that%
\begin{equation}
D_{F}( \overline{A};B) =\label{eq:discord-recovery}
-\log\sup_{\substack{\mathcal{U}_{A\rightarrow XE}^{\mathcal{M}}%
,\\\mathcal{R}_{X\rightarrow XE}}}F\left(  \mathcal{U}_{A\rightarrow
XE}^{\mathcal{M}}( \rho_{AB}) ,\mathcal{R}_{X\rightarrow XE}\left(
\mathcal{M}_{A\rightarrow X}( \rho_{AB}) \right)  \right)  ,
\end{equation}
where $\mathcal{M}_{A\rightarrow X}$ is defined in \eqref{eq:meas-map},
$U_{A\rightarrow XE}^{\mathcal{M}}$ is defined in \eqref{eq:meas-isometry},
and $\mathcal{U}_{A\rightarrow XE}^{\mathcal{M}}$ is defined in \eqref{eq:iso-meas-map}.
\end{definition}

This quantity has a similar interpretation as the original discord, as
summarized in the following quote from \cite{zurek}:

\begin{quote}
\textquotedblleft A vanishing discord can be considered as an indicator of the
superselection rule, or --- in the case of interest --- its value is a measure
of the efficiency of einselection. When [the discord] is large for any
measurement, a lot of information is missed and destroyed by any measurement
on the apparatus alone, but when [the discord] is small almost all the
information about [the system] that exists in the [system--apparatus]
correlations is locally recoverable from the state of the
apparatus.\textquotedblright
\end{quote}

Indeed, we can rewrite $D_{F}$ as characterizing how well a bipartite state
$\rho_{AB}$ is preserved when an entanglement-breaking channel \cite{HSR03}%
\ acts on the $A$ system:

\begin{proposition}
\label{prop:discord-rewrite}For a bipartite state $\rho_{AB}$, we have the
following equality:%
\begin{equation}
D_{F}( \overline{A};B) =-\log\sup_{\mathcal{E}_{A}}F\left(  \rho
_{AB},\mathcal{E}_{A}( \rho_{AB}) \right)  ,
\end{equation}
where the optimization on the right-hand side is over the convex set of
entanglement-breaking channels acting on the system $A$.
\end{proposition}

\begin{proof}
We begin by establishing that%
\begin{equation}
\sup_{\substack{\mathcal{U}_{A\rightarrow XE}^{\mathcal{M}},\mathcal{R}%
_{X\rightarrow XE}}}F\left(  \mathcal{U}_{A\rightarrow XE}^{\mathcal{M}}(
\rho_{AB}) ,\mathcal{R}_{X\rightarrow XE}\left(  \mathcal{M}_{A\rightarrow X}(
\rho_{AB}) \right)  \right) \\
\leq\sup_{\mathcal{E}_{A}}F\left(  \rho_{AB},\mathcal{E}_{A}\left(  \rho
_{AB}\right)  \right)  .
\end{equation}
Let $\mathcal{M}_{A\rightarrow X}$ be any measurement map, let
$U_{A\rightarrow XE}^{\mathcal{M}}$ be an isometric extension for it, and let
$\mathcal{R}_{X\rightarrow XE}$ be any recovery map. Let $\mathcal{T}%
_{XE\rightarrow A}$ denote the following quantum channel:%
\begin{equation}
\mathcal{T}_{XE\rightarrow A}\left(  \gamma_{XE}\right)  \equiv\left(
U^{\mathcal{M}}\right)  ^{\dag}\gamma_{XE}U^{\mathcal{M}}
+\text{Tr}\left\{  \left(  I-U^{\mathcal{M}}\left(  U^{\mathcal{M}}\right)
^{\dag}\right)  \gamma_{XE}\right\}  \sigma_{A},
\end{equation}
where $\sigma_{A}$ is some state on the system $A$. Observe that%
\begin{equation}
\left(  \mathcal{T}_{XE\rightarrow A}\circ\mathcal{U}_{A\rightarrow
XE}^{\mathcal{M}}\right)  ( \rho_{AB}) =\rho_{AB}.
\end{equation}
Then consider that 
\begin{align}
&  F\left(  \mathcal{U}_{A\rightarrow XE}^{\mathcal{M}}\left(  \rho
_{AB}\right)  ,\mathcal{R}_{X\rightarrow XE}\left(  \mathcal{M}_{A\rightarrow
X}(  \rho_{AB})  \right)  \right)  \nonumber\\
&  \leq F\left(  \mathcal{T}_{XE\rightarrow A}\left(  \mathcal{U}%
_{A\rightarrow XE}^{\mathcal{M}}(  \rho_{AB})  \right)
,\mathcal{T}_{XE\rightarrow A}\left(  \mathcal{R}_{X\rightarrow XE}\left(
\mathcal{M}_{A\rightarrow X}(  \rho_{AB})  \right)  \right)
\right)  \\
&  =F\left(  \rho_{AB},\mathcal{T}_{XE\rightarrow A}\left(  \mathcal{R}%
_{X\rightarrow XE}\left(  \mathcal{M}_{A\rightarrow X}\left(  \rho
_{AB}\right)  \right)  \right)  \right)  \\
&  \leq\sup_{\mathcal{E}_{A}}F\left(  \rho_{AB},\mathcal{E}_{A}\left(
\rho_{AB}\right)  \right)  .
\end{align}

The first inequality is a consequence of the monotonicity of fidelity with
respect to quantum operations and the last follows because any entanglement
breaking channel can be written as a concatenation of a measurement followed
by a preparation. In the third line, the measurement is $\mathcal{M}%
_{A\rightarrow X}$ and the preparation is $\mathcal{T}_{XE\rightarrow A}%
\circ\mathcal{R}_{X\rightarrow XE}$.

We now prove the other inequality:%
\begin{equation}
\sup_{\substack{U_{A\rightarrow XE}^{\mathcal{M}},\mathcal{R}_{X\rightarrow
XE}}}F\left(  \mathcal{U}_{A\rightarrow XE}^{\mathcal{M}}\left(  \rho
_{AB}\right)  ,\mathcal{R}_{X\rightarrow XE}\left(  \mathcal{M}_{A\rightarrow
X}( \rho_{AB}) \right)  \right) \label{eq:discord-rewrite}\\
\geq\sup_{\mathcal{E}_{A}}F\left(  \rho_{AB},\mathcal{E}_{A}\left(  \rho
_{AB}\right)  \right)  .
\end{equation}
Let $\mathcal{E}_{A}$ be any entanglement-breaking channel, which consists of
a measurement $\mathcal{M}_{A\rightarrow X}$ followed by a preparation
$\mathcal{P}_{X\rightarrow A}$. Let $U_{A\rightarrow XE}^{\mathcal{M}}$ be an
isometric extension of the measurement map. Then consider that%
\begin{align}
  F( \rho_{AB},\mathcal{E}_{A}( \rho_{AB}) )
&  =F\left(  \rho_{AB},\mathcal{P}_{X\rightarrow A}\left(  \mathcal{M}%
_{A\rightarrow X}( \rho_{AB}) \right)  \right) \\
&  =F\left(  \mathcal{U}_{A\rightarrow XE}^{\mathcal{M}}\left(  \rho
_{AB}\right)  ,\mathcal{U}_{A\rightarrow XE}^{\mathcal{M}}\left(
\mathcal{P}_{X\rightarrow A}\left(  \mathcal{M}_{A\rightarrow X}\left(
\rho_{AB}\right)  \right)  \right)  \right) \\
&  \leq\sup_{\substack{U_{A\rightarrow XE}^{\mathcal{M}},\\\mathcal{R}%
_{X\rightarrow XE}}}F\left(  \mathcal{U}_{A\rightarrow XE}^{\mathcal{M}}(
\rho_{AB}) ,\mathcal{R}_{X\rightarrow XE}\left(  \mathcal{M}_{A\rightarrow X}(
\rho_{AB}) \right)  \right)  ,
\end{align}
where the inequality follows because $\mathcal{U}_{A\rightarrow XE}%
^{\mathcal{M}}\circ\mathcal{P}_{X\rightarrow A}$ is a particular recovery map.
So (\ref{eq:discord-rewrite}) follows and this concludes the proof.
\end{proof}

The proof follows the interpretation given in the quote above:\ the
measurement map $\mathcal{M}_{A\rightarrow X}$ is performed on the $A$ system
of the state $\rho_{AB}$, which is followed by a recovery map $\mathcal{P}%
_{X\rightarrow A}$ that attempts to recover the $A$ system from the state of
the measuring apparatus. Since the measurement map has a classical output, any
recovery map acting on such a classical system is equivalent to a preparation
map. So the quantity $D_{F}( \overline{A};B) $ captures how difficult it is to
recover the full bipartite state after some measurement is performed on it,
following the original spirit of the quantum discord. However, the quantity
$D_{F}( \overline{A};B) $ defined above has the advantage of being a
\textquotedblleft one-shot\textquotedblright\ measure, given that the fidelity
has a clear operational meaning in a \textquotedblleft
one-shot\textquotedblright\ setting. If $D_{F}\left(  \overline{A};B\right)  $
is near to zero, then $F\left(  \rho_{AB},\left(  \mathcal{P}_{X\rightarrow
A}\left(  \mathcal{M}_{A\rightarrow X}\left(  \rho_{AB}\right)  \right)
\right)  \right)  $ is close to one, so that it is possible to recover the
system $A$ by performing a recovery map on the state of the apparatus.
Conversely, if $D_{F}( \overline{A};B) $ is far from zero, then the
measurement recoverability is far from one, so that it is not possible to
recover system $A$ from the state of the measuring apparatus.

The observation in Proposition~\ref{prop:discord-rewrite} leads to the
following proposition, which characterizes quantum states with discord nearly
equal to zero.

\begin{proposition}
[Approximate faithfulness]\label{prop:approx-faithful}A bipartite quantum
state $\rho_{AB}$ has quantum discord nearly equal to zero if and only if it
is an approximate fixed point of an entanglement breaking channel. More
precisely, we have the following:\ If there exists an entanglement breaking
channel $\mathcal{E}_{A}$ and $\varepsilon\in\left[  0,1\right]  $ such that%
\begin{equation}
\left\Vert \rho_{AB}-\mathcal{E}_{A}( \rho_{AB}) \right\Vert _{1}%
\leq\varepsilon, \label{eq:first-impl-1}%
\end{equation}
then the quantum discord $D( \overline{A};B) _{\rho}$\ obeys the following
bound%
\begin{equation}
D( \overline{A};B) _{\rho}\leq4h_{2}( \varepsilon) +8\varepsilon\log\left\vert
A\right\vert , \label{eq:first-impl-2}%
\end{equation}
where $h_{2}( \varepsilon) $ is the binary entropy with the property that
$\lim_{\varepsilon\searrow0}h_{2}( \varepsilon) =0$. Conversely, if the
quantum discord $D( \overline{A};B) _{\rho}$ obeys the following bound for
$\varepsilon\in\left[  0,1\right]  $:%
\begin{equation}
D( \overline{A};B) _{\rho}\leq\varepsilon, \label{eq:second-impl-1}%
\end{equation}
then there exists an entanglement breaking channel $\mathcal{E}_{A}$ such that%
\begin{equation}
\left\Vert \rho_{AB}-\mathcal{E}_{A}( \rho_{AB}) \right\Vert _{1}\leq
2\sqrt{\varepsilon}. \label{eq:second-impl-2}%
\end{equation}

\end{proposition}

\begin{proof}
We begin by proving (\ref{eq:first-impl-1})-(\ref{eq:first-impl-2}). Since any
entanglement breaking channel $\mathcal{E}_{A}$\ consists of a measurement map
$\mathcal{M}_{A\rightarrow X}$\ followed by a preparation map $\mathcal{P}%
_{X\rightarrow A}$, we can write $\mathcal{E}_{A}=\mathcal{P}_{X\rightarrow
A}\circ\mathcal{M}_{A\rightarrow X}$. Then consider that%
\begin{align}
D( \overline{A};B) _{\rho}  &  =I( A;B) _{\rho}-\sup_{\left\{  \Lambda
^{x}\right\}  }I\left(  X;B\right)  _{\sigma}\\
&  \leq I( A;B) _{\rho}-I\left(  X;B\right)  _{\mathcal{M}\left(  \rho\right)
}\\
&  \leq I( A;B) _{\rho}-I( A;B) _{\mathcal{P}\circ\mathcal{M}\left(
\rho\right)  }\\
&  =I( A;B) _{\rho}-I( A;B) _{\mathcal{E}\left(  \rho\right)  }\\
&  \leq4h_{2}( \varepsilon) +8\varepsilon\log\left\vert A\right\vert .
\end{align}
The first inequality follows because the measurement given by $\mathcal{M}%
_{A\rightarrow X}$ is not necessarily optimal. The second inequality is a
consequence of the quantum data processing inequality, in which quantum mutual
information is non-increasing with respect to the local operation
$\mathcal{P}_{X\rightarrow A}$. The last equality follows because
$\mathcal{E}_{A}=\mathcal{P}_{X\rightarrow A}\circ\mathcal{M}_{A\rightarrow
X}$. The last inequality is a consequence of the Alicki-Fannes inequality
\cite{AF04}.

We now prove (\ref{eq:second-impl-1})-(\ref{eq:second-impl-2}). The
Fawzi-Renner inequality $I(A;B|C)_{\rho}\geq-\log F(A;B|C)_{\rho}$ which holds
for any tripartite state $\rho_{ABC}$ \cite{FR14}, combined with other
observations recalled in this section connecting discord with conditional
mutual information, gives us that there exists an entanglement breaking
channel $\mathcal{E}_{A}$ such that%
\begin{align}
D( \overline{A};B) _{\rho}  &  \geq-\log F\left(  \rho_{AB},\mathcal{E}_{A}(
\rho_{AB}) \right) \\
&  \geq-\log\left(  1-\frac{1}{4}\left\Vert \rho_{AB}-\mathcal{E}_{A}\left(
\rho_{AB}\right)  \right\Vert _{1}^{2}\right) \\
&  \geq\frac{1}{4}\left\Vert \rho_{AB}-\mathcal{E}_{A}\left(  \rho
_{AB}\right)  \right\Vert _{1}^{2},
\end{align}
where the second inequality follows from well known relations between trace
distance and fidelity \cite{FG98}\ and the last from $-\log\left(  1-x\right)
\geq x$, valid for $x\leq1$. This is sufficient to conclude
(\ref{eq:second-impl-1})-(\ref{eq:second-impl-2}).
\end{proof}

\begin{remark}
The main conclusion we can take from Proposition~\ref{prop:approx-faithful} is
that quantum states with discord nearly equal to zero are such that they are
recoverable after performing some measurement on one share of them, making
precise the quote from \cite{zurek} given above. In prior work \cite[Lemma
8.12]{H06}, quantum states with discord exactly equal to zero were
characterized as being entirely classical on the system being measured, but
this condition is perhaps too restrictive for characterizing states with
discord approximately equal to zero.
\end{remark}

\begin{remark}
In prior work, discord-like measures of the following form have been widely
considered throughout the literature \cite{KBCPV12}:%
\begin{align}
&  \inf_{\chi_{AB}\in\text{CQ}}\Delta\left(  \rho_{AB},\chi_{AB}\right)  ,\\
&  \inf_{\chi_{AB}\in\text{CC}}\Delta\left(  \rho_{AB},\chi_{AB}\right)  ,
\end{align}
where CQ\ and CC are the respective sets of classical-quantum and
classical-classical states and $\Delta$ is some suitable (pseudo-)distance
measure such as relative entropy, trace distance, or Hilbert-Schmidt distance.
The larger message of Proposition~\ref{prop:approx-faithful} is that it seems
more reasonable from the physical perspective argued in this section and in
the original discord paper \cite{zurek}\ to consider discord-like measures of
the following form:%
\begin{align}
&  \inf_{\mathcal{E}_{A}}\Delta\left(  \rho_{AB},\mathcal{E}_{A}\left(
\rho_{AB}\right)  \right)  ,\\
&  \inf_{\mathcal{E}_{A},\mathcal{E}_{B}}\Delta\left(  \rho_{AB},\left(
\mathcal{E}_{A}\otimes\mathcal{E}_{B}\right)  ( \rho_{AB}) \right)  ,
\end{align}
where the optimization is over the convex set of entanglement breaking
channels and $\Delta$ is again some suitable (pseudo-)distance measure as
mentioned above. One can understand these measures as being a special case of
the proposed measures in \cite{Piani2014}, but we stress here that we arrived
at them independently through the line of reasoning given in this section.
\end{remark}

We now establish some properties of the surprisal of measurement recoverability:

\begin{proposition}
[Local isometric invariance]$D_{F}( \overline{A};B) _{\rho}$ is invariant with
respect to local isometries, in the sense that%
\begin{equation}
D_{F}( \overline{A};B) _{\rho}=D_{F}(\overline{A^{\prime}};B^{\prime}%
)_{\sigma},
\end{equation}
where%
\begin{equation}
\sigma_{A^{\prime}B^{\prime}}\equiv\left(  \mathcal{U}_{A\rightarrow
A^{\prime}}\otimes\mathcal{V}_{B\rightarrow B^{\prime}}\right)  \left(
\rho_{AB}\right)
\end{equation}
and $\mathcal{U}_{A\rightarrow A^{\prime}}$ and $\mathcal{V}_{B\rightarrow
B^{\prime}}$ are isometric CPTP\ maps.
\end{proposition}

\begin{proof}
Let $\mathcal{E}_{A}$ be some entanglement-breaking channel. Let
$\mathcal{T}_{A^{\prime}\rightarrow A}^{\mathcal{U}}$ and $\mathcal{T}%
_{B^{\prime}\rightarrow B}^{\mathcal{V}}$ denote the CPTP\ maps defined in
(\ref{eq:T-maps}). Then from invariance of fidelity with respect to isometries
and the identities in (\ref{eq:invert-isometry-A})-(\ref{eq:invert-isometry-B}%
), we find that 
\begin{align}
&  F(  \rho_{AB},\mathcal{E}_{A}(  \rho_{AB})  )
\nonumber\\
&  =F\left(  \left(  \mathcal{U}_{A\rightarrow A^{\prime}}\otimes
\mathcal{V}_{B\rightarrow B^{\prime}}\right)  (  \rho_{AB})
,\left(  \mathcal{U}_{A\rightarrow A^{\prime}}\otimes\mathcal{V}_{B\rightarrow
B^{\prime}}\right)  \left(  \mathcal{E}_{A}(  \rho_{AB})  \right)
\right) \\
&  =F\left(  \left(  \mathcal{U}_{A\rightarrow A^{\prime}}\otimes
\mathcal{V}_{B\rightarrow B^{\prime}}\right)  (  \rho_{AB})
,\left(  \mathcal{U}_{A\rightarrow A^{\prime}}\circ\mathcal{E}_{A}%
\circ\mathcal{T}_{A^{\prime}\rightarrow A}^{\mathcal{U}}\right)  \left[
\left(  \mathcal{U}_{A\rightarrow A^{\prime}}\otimes\mathcal{V}_{B\rightarrow
B^{\prime}}\right)  (  \rho_{AB})  \right]  \right) \\
&  \leq\sup_{\mathcal{E}_{A^{\prime}}}F\left(  \left(  \mathcal{U}%
_{A\rightarrow A^{\prime}}\otimes\mathcal{V}_{B\rightarrow B^{\prime}}\right)
(  \rho_{AB})  ,\mathcal{E}_{A^{\prime}}\left(  \left(
\mathcal{U}_{A\rightarrow A^{\prime}}\otimes\mathcal{V}_{B\rightarrow
B^{\prime}}\right)  (  \rho_{AB})  \right)  \right)  .
\end{align}
Since the inequality is true for any entanglement breaking channel
$\mathcal{E}_{A}$, we find after applying a negative logarithm that%
\begin{equation}
D_{F}(  \overline{A};B)  _{\rho}\geq D_{F}\left(  \overline
{A};B\right)  _{\left(  \mathcal{U}\otimes\mathcal{V}\right)  \left(
\rho\right)  }.
\end{equation}
Now consider that%
\begin{align}
&  F\left(  \left(  \mathcal{U}_{A\rightarrow A^{\prime}}\otimes
\mathcal{V}_{B\rightarrow B^{\prime}}\right)  (  \rho_{AB})
,\mathcal{E}_{A^{\prime}}\left[  \left(  \mathcal{U}_{A\rightarrow A^{\prime}%
}\otimes\mathcal{V}_{B\rightarrow B^{\prime}}\right)  \left(  \rho
_{AB}\right)  \right]  \right) \nonumber\\
&  =F\left(  \mathcal{U}_{A\rightarrow A^{\prime}}(  \rho_{AB})
,\left(  \mathcal{E}_{A^{\prime}}\circ\mathcal{U}_{A\rightarrow A^{\prime}%
}\right)  (  \rho_{AB})  \right) \\
&  \leq F\left(  \left(  \mathcal{T}_{A^{\prime}\rightarrow A}^{\mathcal{U}%
}\circ\mathcal{U}_{A\rightarrow A^{\prime}}\right)  (  \rho_{AB})
,\left(  \mathcal{T}_{A^{\prime}\rightarrow A}^{\mathcal{U}}\circ
\mathcal{E}_{A^{\prime}}\circ\mathcal{U}_{A\rightarrow A^{\prime}}\right)
(  \rho_{AB})  \right) \\
&  =F\left(  \rho_{AB},\left(  \mathcal{T}_{A^{\prime}\rightarrow
A}^{\mathcal{U}}\circ\mathcal{E}_{A^{\prime}}\circ\mathcal{U}_{A\rightarrow
A^{\prime}}\right)  (  \rho_{AB})  \right) \\
&  \leq\sup_{\mathcal{E}_{A}}F\left(  \rho_{AB},\mathcal{E}_{A}\left(
\rho_{AB}\right)  \right)  .
\end{align}

Since the inequality is true for any entanglement breaking channel
$\mathcal{E}_{A^{\prime}}$, we find after applying a negative logarithm that%
\begin{equation}
D_{F}( \overline{A};B) _{\rho}\leq D_{F}\left(  \overline{A};B\right)
_{\left(  \mathcal{U}\otimes\mathcal{V}\right)  \left(  \rho\right)  },
\end{equation}
which gives the statement of the proposition.
\end{proof}

\begin{proposition}
[Exact faithfulness]The surprisal of measurement recoverability $D_{F}\left(
\overline{A};B\right)  _{\rho}$\ is equal to zero if and only if $\rho_{AB}$
is a classical-quantum state, having the form%
\begin{equation}
\rho_{AB}=\sum_{x}p_{X}( x) \left\vert x\right\rangle \left\langle
x\right\vert _{A}\otimes\rho_{B}^{x},
\end{equation}
for some orthonormal basis $\left\{  \left\vert x\right\rangle \right\}  $,
probability distribution $p_{X}( x) $, and states $\left\{  \rho_{B}%
^{x}\right\}  $.
\end{proposition}

\begin{proof}
Suppose that the state is classical-quantum. Then it is a fixed point of the
entanglement breaking map $\sum_{x}\left\vert x\right\rangle \left\langle
x\right\vert _{A}\left(  \cdot\right)  \left\vert x\right\rangle \left\langle
x\right\vert _{A}$, so that the fidelity of measurement recovery is equal to
one and its surprisal is equal to zero. On the other hand, suppose that
$D_{F}( \overline{A};B) _{\rho}=0$. Then this means that there exists an
entanglement breaking channel $\mathcal{E}_{A}$ of which $\rho_{AB}$ is a
fixed point (since $F\left(  \rho_{AB},\mathcal{E}_{A}\left(  \rho
_{AB}\right)  \right)  =1$ is equivalent to $\rho_{AB}=\mathcal{E}_{A}\left(
\rho_{AB}\right)  $), and furthermore, applying the fixed point projection%
\begin{equation}
\overline{\mathcal{E}_{A}}\equiv\lim_{K\rightarrow\infty}\frac{1}{K}\sum
_{k=1}^{K}\mathcal{E}_{A}^{k}%
\end{equation}
leaves $\rho_{AB}$ invariant. The map $\overline{\mathcal{E}_{A}}$ has been
characterized in \cite[Theorem 5.3]{FNW14} to be an entanglement breaking
channel of the following form:%
\begin{equation}
\overline{\mathcal{E}_{A}}\left(  \cdot\right)  =\sum_{x}\operatorname{Tr}%
\left\{  \Lambda^{x}_A\left(  \cdot\right)  \right\}  \sigma^{x}_A,
\end{equation}
where the states $\sigma^{x}_A$ have orthogonal support, $\Lambda^{x}_A\geq0$, and
$\sum_{x}\Lambda^{x}_A=I$. Applying this channel to $\rho_{AB}$ then gives a
classical-quantum state, and since $\rho_{AB}$ is invariant with respect to
the action of this channel to begin with, it must have been classical-quantum
from the start.
\end{proof}

\begin{proposition}
[Dimension bound]\label{prop:disc-dim-bound}The surprisal of measurement
recoverability obeys the following dimension bound:%
\begin{equation}
D_{F}( \overline{A};B) _{\rho}\leq\log\left\vert A\right\vert ,
\end{equation}
or equivalently,%
\begin{equation}
\sup_{\mathcal{E}_{A}}F\left(  \rho_{AB},\mathcal{E}_{A}\left(  \rho
_{AB}\right)  \right)  \geq\frac{1}{\left\vert A\right\vert }.
\end{equation}

\end{proposition}

\begin{proof}
The idea behind the proof is to consider an entanglement breaking channel
$\mathcal{E}_{A}$ that completely dephases the system $A$. Let $\overline
{\Delta}_{A}$ denote such a channel, so that%
\begin{equation}
\overline{\Delta}_{A}\left(  \cdot\right)  \equiv\sum_{i}\left\vert
i\right\rangle \left\langle i\right\vert _{A}\left(  \cdot\right)  \left\vert
i\right\rangle \left\langle i\right\vert _{A},
\end{equation}
where $\left\{  \left\vert i\right\rangle _{A}\right\}  $ is some orthonormal
basis spanning the space for the $A$ system. Let a spectral decomposition of
$\rho_{AB}$ be given by%
\begin{equation}
\rho_{AB}=\sum_{x}p_{X}( x) \left\vert \psi^{x}\right\rangle \left\langle
\psi^{x}\right\vert _{AB},
\end{equation}
where $p_{X}$ is a probability distribution and $\left\{  \left\vert \psi
^{x}\right\rangle _{AB}\right\}  $ is a set of pure states. We then find that%
\begin{align}
  D_{F}( \overline{A};B) _{\rho}
&  \leq-\log F\left(  \rho_{AB},\overline{\Delta}_{A}( \rho_{AB}) \right) \\
&  =-2\log\sqrt{F}\left(  \rho_{AB},\overline{\Delta}_{A}\left(  \rho
_{AB}\right)  \right) \\
&  \leq\sum_{x}p_{X}( x) \left[  -2\log\sqrt{F}\left(  \psi_{AB}^{x}%
,\overline{\Delta}_{A}\left(  \psi_{AB}^{x}\right)  \right)  \right] \\
&  =\sum_{x}p_{X}( x) \left[  -\log\left\langle \psi^{x}\right\vert
_{AB}\overline{\Delta}_{A}\left(  \psi_{AB}^{x}\right)  \left\vert \psi
^{x}\right\rangle _{AB}\right] \\
&  =\sum_{x}p_{X}( x) \left[  -\log\sum_{i}\left[  \left\langle i\right\vert
_{A}\psi_{A}^{x}\left\vert i\right\rangle _{A}\right]  ^{2}\right] \\
&  \leq\log\left\vert A\right\vert .
\end{align}
The second inequality follows from joint concavity of the root fidelity
$\sqrt{F}$ and convexity of $-\log$. The last equality is a consequence of a
well known expression for the entanglement fidelity of a channel (see, e.g.,
\cite[Theorem~9.5.1]{W13}). The last inequality follows by recognizing%
\begin{equation}
-\log\sum_{i}\left[  \left\langle i\right\vert _{A}\psi_{A}^{x}\left\vert
i\right\rangle _{A}\right]  ^{2}%
\end{equation}
as the R\'{e}nyi 2-entropy of the probability distribution $\left\langle
i\right\vert _{A}\psi_{A}^{x}\left\vert i\right\rangle _{A}$ and from the fact
that all R\'{e}nyi entropies are bounded from above by the logarithm of the
alphabet size of the distribution, which in this case is $\log\left\vert
A\right\vert $.
\end{proof}

Given that the R\'{e}nyi 2-entropy of the marginal of a bipartite pure state
is an entanglement measure, the following proposition demonstrates that the
surprisal of measurement recoverability reduces to an entanglement measure
when evaluated for pure states.

\begin{proposition}
[Pure states]\label{prop:discord-pure-states}Let $\psi_{AB}$ be a pure state.
Then%
\begin{equation}
D_{F}( \overline{A};B) _{\psi}=-\log\operatorname{Tr}\left\{  \psi_{A}%
^{2}\right\}  .
\end{equation}

\end{proposition}

\begin{proof}
For a pure state $\psi_{AB}$, consider that%
\begin{align}
  D_{F}( \overline{A};B) _{\psi}
&  =-\log\sup_{\mathcal{E}_{A}}F\left(  \psi_{AB},\mathcal{E}_{A}\left(
\psi_{AB}\right)  \right) \\
&  =-\log\sup_{\substack{\left\vert \phi_{x}\right\rangle ,\left\vert
\varphi_{x}\right\rangle :\\\left\Vert \left\vert \phi_{x}\right\rangle
\right\Vert _{2}=1,\\\sum_{x}\left\vert \varphi_{x}\right\rangle \left\langle
\varphi_{x}\right\vert =I}}\sum_{x}\left\vert \left\langle \varphi
_{x}\right\vert _{A}\psi_{A}\left\vert \phi_{x}\right\rangle _{A}\right\vert
^{2},
\end{align}
where the optimization in the second line is over pure-state vectors
$\left\vert \phi_{x}\right\rangle $ and corresponding measurement vectors
$\left\vert \varphi_{x}\right\rangle $ satisfying $\sum_{x}\left\vert
\varphi_{x}\right\rangle \left\langle \varphi_{x}\right\vert =I$. The second
equality follows from the formula for entanglement fidelity (see, e.g.,
\cite[Theorem~9.5.1]{W13}) and the fact that the Kraus operators of an
entanglement-breaking channel have the special form $\left\{  \left\vert
\phi_{x}\right\rangle \left\langle \varphi_{x}\right\vert \right\}  _{x}$ with
$\left\vert \phi_{x}\right\rangle $ pure quantum states and $\sum
_{x}\left\vert \varphi_{x}\right\rangle \left\langle \varphi_{x}\right\vert
=I$ \cite{HSR03}. Consider for all such choices, we have that%
\begin{align}
  \sum_{x}\left\vert \left\langle \varphi_{x}\right\vert _{A}\psi
_{A}\left\vert \phi_{x}\right\rangle _{A}\right\vert ^{2}
&  =\sum_{x}\left\langle \varphi_{x}\right\vert _{A}\psi_{A}\left\vert
\phi_{x}\right\rangle \left\langle \phi_{x}\right\vert _{A}\psi_{A}\left\vert
\varphi_{x}\right\rangle _{A}\\
&  \leq\sum_{x}\left\langle \varphi_{x}\right\vert _{A}\psi_{A}^{2}\left\vert
\varphi_{x}\right\rangle _{A}\\
&  =\sum_{x}\text{Tr}\left\{  \left\vert \varphi_{x}\right\rangle \left\langle
\varphi_{x}\right\vert _{A}\psi_{A}^{2}\right\} \\
&  =\text{Tr}\left\{  \psi_{A}^{2}\right\}  ,
\end{align}
where the inequality follows from the operator inequality $\left\vert \phi
_{x}\right\rangle \left\langle \phi_{x}\right\vert _{A}\leq I_{A}$. However, a
particular choice of Kraus operators $\left\{  \left\vert \phi_{x}%
\right\rangle \left\langle \varphi_{x}\right\vert \right\}  _{x}$ is $\left\{
\left\vert \psi^{x}\right\rangle \left\langle \psi^{x}\right\vert \right\}
_{x}$, where $\left\{  \left\vert \psi^{x}\right\rangle \right\}  _{x}$ is the
set of eigenvectors of $\psi_{A}$. For this choice, we find that%
\begin{equation}
\sum_{x}\left\vert \left\langle \psi^{x}\right\vert _{A}\psi_{A}\left\vert
\psi^{x}\right\rangle _{A}\right\vert ^{2}=\text{Tr}\left\{  \psi_{A}%
^{2}\right\}  ,
\end{equation}
so that we can conclude that%
\begin{equation}
\sup_{\left\vert \phi_{x}\right\rangle ,\left\vert \varphi_{x}\right\rangle
:\sum_{x}\left\vert \varphi_{x}\right\rangle \left\langle \varphi
_{x}\right\vert =I}\sum_{x}\left\vert \left\langle \varphi_{x}\right\vert
_{A}\psi_{A}\left\vert \phi_{x}\right\rangle _{A}\right\vert ^{2}
=\text{Tr}\left\{  \psi_{A}^{2}\right\}  .
\end{equation}

\end{proof}

\begin{proposition}
[Normalization]The surprisal of measurement recoverability $D_{F}\left(
\overline{A};B\right)  _{\Phi}$ is equal to $\log d$ for a maximally entangled
state $\Phi_{AB}$ with Schmidt rank $d$.
\end{proposition}

\begin{proof}
This is a direct consequence of Proposition~\ref{prop:discord-pure-states} and
the fact that $\Phi_{A}=I_{A}/d$.
\end{proof}

\begin{proposition}
[Monotonicity]The surprisal of measurement recoverability is monotone with
respect to quantum operations on the unmeasured system, i.e.,%
\begin{equation}
D_{F}( \overline{A};B) _{\rho}\geq D_{F}\left(  \overline{A};B^{\prime
}\right)  _{\sigma},
\end{equation}
where $\sigma_{AB^{\prime}}\equiv\mathcal{N}_{B\rightarrow B^{\prime}}\left(
\rho_{AB}\right)  $.
\end{proposition}

\begin{proof}
Intuitively, this follows because it is easier to recover from a measurement
when the state is noisier to begin with. Indeed, let $\mathcal{E}_{A}$ be an
entanglement breaking channel. Then%
\begin{align}
F( \rho_{AB},\mathcal{E}_{A}( \rho_{AB}) )  &  \leq F\left(  \sigma
_{AB^{\prime}},\mathcal{E}_{A}\left(  \sigma_{AB^{\prime}}\right)  \right) \\
&  \leq\sup_{\mathcal{E}_{A}}F\left(  \sigma_{AB^{\prime}},\mathcal{E}%
_{A}\left(  \sigma_{AB^{\prime}}\right)  \right)  ,
\end{align}
where the first inequality is due to the fact that $\mathcal{E}_{A}$ commutes
with $\mathcal{N}_{B\rightarrow B^{\prime}}$ and monotonicity of the fidelity
with respect to quantum channels. Since the inequality holds for all
entanglement breaking channels, we can conclude that%
\begin{equation}
\sup_{\mathcal{E}_{A}}F\left(  \rho_{AB},\mathcal{E}_{A}\left(  \rho
_{AB}\right)  \right)  \leq\sup_{\mathcal{E}_{A}}F\left(  \sigma_{AB^{\prime}%
},\mathcal{E}_{A}\left(  \sigma_{AB^{\prime}}\right)  \right)  .
\end{equation}
Taking a negative logarithm gives the statement of the proposition.
\end{proof}

With a proof nearly identical to that for Proposition~\ref{prop:geo-SE-cont},
we find that $D_{F}( \overline{A};B) _{\rho}$ is continuous:

\begin{proposition}
[Continuity]$D_{F}( \overline{A};B) $ is a continuous function of its input.
That is, given two bipartite states $\rho_{AB}$ and $\sigma_{AB}$ such that
$F\left(  \rho_{AB},\sigma_{AB}\right)  \geq1-\varepsilon$ where
$\varepsilon\in\left[  0,1\right]  $, then the following inequalities hold%
\begin{align}
\left\vert \sup_{\mathcal{E}_{A}}F\left(  \rho_{AB},\mathcal{E}_{A}\left(
\rho_{AB}\right)  \right)  -\sup_{\mathcal{E}_{A}}F\left(  \sigma
_{AB},\mathcal{E}_{A}\left(  \sigma_{AB}\right)  \right)  \right\vert 
& \leq8\sqrt{\varepsilon}, \\
\left\vert D_{F}( \overline{A};B) _{\rho}-D_{F}\left(  \overline{A};B\right)
_{\sigma}\right\vert & \leq\left\vert A\right\vert 8\sqrt{\varepsilon}.
\end{align}

\end{proposition}

\section{Multipartite fidelity of recovery}

We state here that it is certainly possible to generalize the fidelity of
recovery to the multipartite setting. Indeed, by following the same line of
reasoning mentioned in the introduction (starting from the R\'{e}nyi
conditional multipartite information \cite[Section~10.1]{BSW14} and
understanding the $\alpha=1/2$ quantity in terms of several Petz recovery
maps), we can define the multipartite fidelity of recovery for a multipartite
state $\rho_{A_{1}\cdots A_{l}C}$ as follows:%
\begin{equation}
F\left(  A_{1};A_{2};\cdots;A_{l}|C\right)  _{\rho}=
\sup_{\substack{\mathcal{R}_{C\rightarrow A_{1}C}^{1},\\\ldots,\\\mathcal{R}%
_{C\rightarrow A_{l-1}C}^{l-1}}}F\left(  \rho_{A_{1}\cdots A_{l}C}%
,\mathcal{R}_{C\rightarrow A_{1}C}^{1}\circ\cdots\circ\mathcal{R}%
_{C\rightarrow A_{l-1}C}^{l-1}\left(  \rho_{A_{l}C}\right)  \right)  .
\end{equation}
The interpretation of this quantity is as written:\ systems $A_{1}$ through
$A_{l-1}$ of the state $\rho_{A_{1}\cdots A_{l}C}$ are lost, and one attempts
to recover them one at a time by performing a sequence of recovery maps on
system $C$ alone. We can then define a quantity analogous to the multipartite
conditional mutual information as follows:%
\begin{equation}
I_{F}( A_{1};A_{2};\cdots;A_{l}|C) _{\rho}\equiv-\log F( A_{1};A_{2}%
;\cdots;A_{l}|C) _{\rho},
\end{equation}
and one can easily show along the lines given for the bipartite case that the
resulting multipartite quantity is non-negative, monotone with respect to
local operations, and obeys a dimension bound.

We leave it as an open question to develop fully a multipartite geometric
squashed entanglement, defined by replacing the conditional multipartite
mutual information in the usual definition \cite{YHHHOS09}\ with $I_{F}$ given
above. One could also explore multipartite versions of the surprisal of
measurement recoverability.

\section{Conclusion}

We have defined the fidelity of recovery $F( A;B|C) _{\rho}$\ of a tripartite
state $\rho_{ABC}$ to quantify how well one can recover the full state on all
three systems if system $A$ is lost and the recovery map can act only on
system $C$. By taking the negative logarithm of the fidelity of recovery, we
obtain an entropic quantity $I_{F}( A;B|C) _{\rho}$ which obeys nearly all of
the entropic relations that the conditional mutual information does. The
quantities $F( A;B|C) _{\rho}$ and $I_{F}( A;B|C) _{\rho}$ are rooted in our
earlier work on seeking out R\'{e}nyi generalizations of the conditional
mutual information \cite{BSW14}. Whereas we have not been able to prove that
all of the aforementioned properties hold for the R\'{e}nyi conditional mutual
informations from \cite{BSW14}, it is pleasing to us that it is relatively
straightforward to show that these properties hold for $I_{F}( A;B|C) _{\rho}$.

Another contribution was to define the geometric squashed entanglement
$E_{F}^{\operatorname{sq}}( A;B) _{\rho}$, inspired by the original squashed
entanglement measure from \cite{CW04}. We proved that $E_{F}%
^{\operatorname{sq}}( A;B) _{\rho}$ is a 1-LOCC monotone, is invariant with
respect to local isometries, is faithful, reduces to the well known geometric
measure of entanglement \cite{WG03,CAH13} when the bipartite state is pure,
normalized on maximally entangled states, subadditive, and continuous. The
geometric squashed entanglement could find applications in \textquotedblleft
one-shot\textquotedblright\ scenarios of quantum information theory, since it
is fundamentally a one-shot measure based on the fidelity. (The fidelity is
said to be a ``one-shot'' quantity because it has an operational meaning in
terms of a single experiment: it is the probability with which a purification
of one state could pass a test for being a purification of the other state.)

Our final contribution was to define the surprisal of measurement
recoverability $D_{F}( \overline{A};B) _{\rho}$, a quantum correlation measure
having physical roots in the same vein as those used to justify the definition
of the quantum discord. We showed that it is non-negative, invariant with
respect to local isometries, faithful on classical-quantum states, obeys a
dimension bound, and is continuous. Furthermore, we used this quantity to
characterize quantum states with discord nearly equal to zero, finding that
such states are approximate fixed points of an entanglement breaking channel.

From here, there are several interesting lines of inquiry to pursue. It is
clear that generally $I_{F}(A;B|C) \neq I_{F}(B;A|C)$: can we quantify how
large the gap can be between them in general? Can we prove a stronger chain
rule for the fidelity of recovery?\ If something along these lines holds, it
might be helpful in establishing that the geometric squashed entanglement is
monogamous or additive. (At the very least, we can say that geometric squashed
entanglement is additive with respect to pure states, given that it reduces to
the geometric measure of entanglement which is clearly additive by inspecting
(\ref{eq:geo-reduce-2}).) Is it possible to improve our continuity bounds to
attain \textquotedblleft asymptotic continuity\textquotedblright? Can one show
that geometric squashed entanglement is nonlockable \cite{C06}? Preliminary
evidence from considering the strongest known locking schemes from
\cite{FHS11} suggests that it might not be lockable. We are also interested in
a multipartite geometric squashed entanglement, but we face similar challenges
as those discussed in \cite{LW14}\ for establishing its faithfulness.

\bigskip

\textbf{Acknowledgements.} We are grateful to Gerardo Adesso, Mario Berta,
Todd Brun, Marco Piani, and Masahiro Takeoka for helpful discussions about
this work. KS acknowledges support from NSF\ Grant No.~CCF-1350397, the DARPA
Quiness Program through US Army Research Office award W31P4Q-12-1-0019, and
the Graduate School of Louisiana State University for the 2014-2015
Dissertation Year Fellowship. MMW acknowledges support from the APS-IUSSTF
Professorship Award in Physics, startup funds from the Department of Physics
and Astronomy at LSU, support from the NSF\ under Award No.~CCF-1350397, and
support from the DARPA Quiness Program through US Army Research Office award W31P4Q-12-1-0019.

\appendix

\section{Appendix}

\label{app}

Given a state $\rho$, a positive semidefinite operator $\sigma$, and
$\alpha\in\lbrack0,1)\cup(1,\infty)$, we define the R\'{e}nyi relative entropy
as%
\begin{equation}
D_{\alpha}(  \rho\Vert\sigma)  \equiv\frac{1}{\alpha-1}%
\log\text{Tr}\left\{  \rho^{\alpha}\sigma^{1-\alpha}\right\}  ,
\end{equation}
whenever the support of $\rho$ is contained in the support of $\sigma$, and it
is equal to $+\infty$ otherwise. The conditional R\'{e}nyi entropy of a
bipartite state $\rho_{AB}$ is defined as%
\begin{equation}
H_{\alpha}(  A|B)  _{\rho}\equiv-D_{\alpha}(  \rho_{AB}\Vert
I_{A}\otimes\rho_{B})  .
\end{equation}
(See, e.g., \cite{TCR09} for details of these definitions.)\ This leads us to
the following lemma:

\begin{lemma}
\label{lem:classical-non-neg}Let $\rho_{XB}$ be a classical-quantum state,
i.e., such that%
\begin{equation}
\rho_{XB}\equiv\sum_{x}p( x) \left\vert x\right\rangle \left\langle
x\right\vert _{X}\otimes\rho_{B}^{x},
\end{equation}
where $p( x) $ is a probability distribution and $\{\rho_{B}^{x}\}$ is a set
of quantum states. For $\alpha\in\lbrack0,1)\cup(1,2]$,%
\begin{equation}
H_{\alpha}\left(  X|B\right)  \geq0.
\end{equation}

\end{lemma}

\begin{proof}
This follows because it is possible to copy classical information, and
conditional entropy increases with respect to the loss of a classical copy.
Consider the following extension of $\rho_{XB}$:%
\begin{equation}
\rho_{X\hat{X}B}\equiv\sum_{x}p( x) \left\vert x\right\rangle \left\langle
x\right\vert _{X}\otimes\left\vert x\right\rangle \left\langle x\right\vert
_{\hat{X}}\otimes\rho_{B}^{x}.
\end{equation}
Then we show that $H_{\alpha}(X|\hat{X}B)=0$ for all $\alpha\in\lbrack
0,1)\cup\left(  1,\infty\right)  $. Indeed, consider that 
\begin{align}
&  H_{\alpha}(X|\hat{X}B)\nonumber\\
&  =\frac{1}{1-\alpha}\log\text{Tr}\left\{  \left(  \sum_{x}p(  x)
\left\vert x\right\rangle \left\langle x\right\vert _{X}\otimes\left\vert
x\right\rangle \left\langle x\right\vert _{\hat{X}}\otimes\rho_{B}^{x}\right)
^{\alpha}\left[  I_{X}\otimes\left(  \sum_{x^{\prime}}p\left(  x^{\prime
}\right)  \left\vert x^{\prime}\right\rangle \left\langle x^{\prime
}\right\vert _{\hat{X}}\otimes\rho_{B}^{x^{\prime}}\right)  ^{1-\alpha
}\right]  \right\} \\
&  =\frac{1}{1-\alpha}\log\text{Tr}\left\{  \sum_{x}p^{\alpha}\left(
x\right)  \left\vert x\right\rangle \left\langle x\right\vert _{X}%
\otimes\left\vert x\right\rangle \left\langle x\right\vert _{\hat{X}}%
\otimes\left(  \rho_{B}^{x}\right)  ^{\alpha}\sum_{x^{\prime}}p^{1-\alpha
}\left(  x^{\prime}\right)  I_{X}\otimes\left\vert x^{\prime}\right\rangle
\left\langle x^{\prime}\right\vert _{\hat{X}}\otimes\left(  \rho
_{B}^{x^{\prime}}\right)  ^{1-\alpha}\right\} \\
&  =\frac{1}{1-\alpha}\log\text{Tr}\left\{  \sum_{x}p(  x)
\left\vert x\right\rangle \left\langle x\right\vert _{X}\otimes\left\vert
x\right\rangle \left\langle x\right\vert _{\hat{X}}\otimes\rho_{B}^{x}\right\}
\\
&  =0.
\end{align}

Then for $\alpha\in\lbrack0,1)\cup(1,2]$, the desired inequality is a
consequence of quantum data processing \cite[Lemma~5]{TCR09}:%
\begin{equation}
H_{\alpha}(X|B)\geq H_{\alpha}(X|\hat{X}B)=0.
\end{equation}

\end{proof}

\bibliographystyle{alpha}
\bibliography{Ref}

\end{document}